\documentclass[letterpaper,twocolumn,10pt]{article}
\usepackage{usenix2019_v3}

\usepackage{tikz}
\usepackage{amsmath,amssymb}
\usepackage{verbatim}
\usepackage{xspace}
\usepackage{amsthm}
\usepackage{textcomp}
\usepackage{booktabs}
\usepackage{caption}
\usepackage{subcaption}
\usepackage{mathtools}
\usepackage{todonotes}
\usepackage{graphicx}
\usepackage{multirow}
\usepackage{tabularx}
\usepackage{bbm}
\usepackage{dsfont}
\usepackage{xcolor}
\usepackage{algcompatible,algorithm}

\newcommand{\ra}[1]{\renewcommand{\arraystretch}{#1}}

\newcommand\numberthis{\addtocounter{equation}{1}\tag{\theequation}}

\newcommand{\name}{{\fontfamily{cmss}\selectfont{Pr$\epsilon\epsilon$ch}}\xspace}

\newtheorem{theorem}{Theorem}[section]

\newtheorem{lemma}[theorem]{Lemma}
\theoremstyle{definition}
\newtheorem{definition}{Definition}[section]
\theoremstyle{assumption}
\newtheorem{assumption}{Assumption}

\newcommand{\CSP}{\textsf{CSP}\xspace}
\newcommand{\OSP}{\textsf{OSP}\xspace}
\newcommand{\CSPs}{\textsf{CSP}s\xspace}

\newcommand{\minrev}[1]{{\color{black} #1}}
\newcommand{\TTS}{\textsf{TTS}\xspace}
\newcommand{\ASR}{\textsf{ASR}\xspace}
\newcommand{\WER}{\textsf{WER}\xspace}
\newcommand{\VC}{\textsf{VC}\xspace}
\newcommand{\CDF}{\textsf{CDF}\xspace}
\newcommand{\SWS}{\textsf{SWS}\xspace}
\newcommand{\NLP}{\textsf{NLP}\xspace}
\newcommand{\NER}{\textsf{NER}\xspace}
\newcommand{\KWS}{\textsf{KWS}\xspace}
\newcommand{\DP}{\textsf{DP}\xspace}

\newcommand{\TPR}{\textsf{TPR}\xspace}
\newcommand{\FPR}{\textsf{FPR}\xspace}
\newcommand{\TFIDF}{\textsf{TF-IDF}\xspace}

\newcommand{\squishlist}{
	\begin{list}{$\bullet$}
		{
			\setlength{\itemsep}{0pt}
			\setlength{\parsep}{3pt}
			\setlength{\topsep}{3pt}
			\setlength{\partopsep}{0pt}
			\setlength{\leftmargin}{1.5em}
			\setlength{\labelwidth}{1em}
			\setlength{\labelsep}{0.5em} } }
	
\newcommand{\squishend}{
\end{list}  }

\usepackage{microtype}
\usepackage{titlesec}

\usepackage{enumitem}
  \titlespacing\paragraph{0pt}{2pt plus 4pt minus 2pt}{4pt plus 2pt minus 2pt}

\newcommand{\squishlistnum}{
	\begin{enumerate}
		{
			\setlength{\itemsep}{0pt}
			\setlength{\parsep}{3pt}
			\setlength{\topsep}{3pt}
			\setlength{\partopsep}{0pt}
			\setlength{\leftmargin}{1.5em}
			\setlength{\labelwidth}{1em}
			\setlength{\labelsep}{0.3em} } }
	
\newcommand{\squishendnum}{
\end{enumerate}  }
\newif\ifpaper
\paperfalse   %

\begin{document}

\title{\Large \bf \name: A System for Privacy-Preserving Speech Transcription}

\author{
{\rm Shimaa Ahmed, Amrita Roy Chowdhury, Kassem Fawaz, and Parmesh Ramanathan}\\
University of Wisconsin-Madison \\
\{ahmed27, roychowdhur2, kfawaz, parmesh.ramanathan\}@wisc.edu
}

\maketitle
\begin{abstract}
New advances in machine learning  
have made Automated Speech Recognition (\ASR) systems practical and more scalable. These systems, however, pose serious privacy threats as speech is a rich source of sensitive acoustic and textual information. 
Although offline and open-source \ASR eliminates the privacy risks, its transcription performance is inferior to that of cloud-based \ASR systems, especially for real-world use cases. In this paper, we propose \name, an end-to-end speech transcription system which lies at an intermediate point in the privacy-utility spectrum. It protects the acoustic features of the speakers’ voices and protects the privacy of the textual content at an improved performance relative to offline \ASR. {Additionally, \name provides several control knobs to allow customizable utility-usability-privacy trade-off.} It relies on cloud-based services to transcribe a speech file after applying a series of privacy-preserving operations on the user’s side. We perform a comprehensive evaluation of \name, using diverse real-world datasets, that demonstrates its effectiveness. \name provides transcription at a 2\% to 32.25\% (mean 17.34\%) relative improvement in word error rate over Deep Speech, while fully obfuscating the speakers' voice biometrics and allowing only a differentially private view of the textual content.
\end{abstract}

\section{Introduction}
\label{sec:intro}

New advances in machine learning and the abundance of speech data have made Automated Speech Recognition (\ASR) systems practical and reliable~\cite{ amodei2016deep, graves2013speech}.  \ASR systems have achieved a near-human performance on standard datasets~\cite{amodei2016deep, graves2013speech}, at a scale. This scalability is desirable in many domains, such as journalism~\cite{mcgregor2015investigating}, law, business, education, and health care, where cost, delay, and third-party legal implications~\cite{nautsch2019gdpr} prohibit the application of manual transcription services~\cite{davino2004assessing}. For example, recent research has identified private voice transcription as one of the challenges journalists face when interviewing sensitive sources~\cite{mcgregor2015investigating}.

Several companies, such as Google and Amazon, provide online APIs for speech transcription.
This convenience, however, comes at the cost of privacy. 
A speech recording contains acoustic features that can reveal sensitive information about the user, such as age, gender~\cite{safavi2018automatic}, emotion~\cite{schuller2013computational, aloufi2019emotionless}, accent, and health conditions~\cite{schuller2013paralinguistics}. The acoustic features are also biometric identifiers of the speakers~\cite{muckenhirn2018towards}, enabling speaker identification and impersonation~\cite{google_speech_synthesis}.  Additionally, the textual content of speech can be sensitive~\cite{nautsch2019gdpr}. For example, medical recordings can contain private health information about patients~\cite{davino2004assessing}, and business recordings can include proprietary information. 
Current cloud services already support several speech processing APIs like speaker identification and diarization. They also support text analysis APIs, such as topic modeling, document categorization, sentiment analysis, and entity detection (Sec.~\ref{sec:textAnalysis}), that can extract sensitive information from text. Applying these APIs to the recorded speech can significantly undermine the user's privacy.

Offline \minrev{and open-source} transcription services, like Deep Speech~\cite{hannun2014deep}, solve 
these privacy challenges
as the speech files never leave the user's trust boundary. 
However,  we find that
their performance does not match that of a cloud service provider~\cite{vaidya2019you}, especially on real-world conversations and different accents (Sec.~\ref{sec:accuracy_comparison}). Thus, the primary goal of this paper is to: \textit{provide an intermediate solution along the utility-privacy spectrum that uses cloud services while providing a formal privacy guarantee.} 

We present \name (Privacy-Preserving Speech) as a means to achieve this goal; it is an end-to-end speech transcription system that: (1) protects the users' privacy along the acoustic and textual dimensions; (2) improves the transcription performance relative to offline \ASR; and (3) provides the user with control knobs to customize the trade-offs between utility, usability, and privacy.

\noindent
\textbf{Textual Privacy:}  \name segments and shuffles the input speech file to break the context of the text, effectively transforming it into a bag-of-words. Then, it injects dummy (noise) segments to provide the formal privacy guarantee of differential privacy (\DP)~\cite{dwork2014algorithmic}. \\
\noindent
\textbf{Acoustic Privacy:}  \name applies voice conversion to protect the acoustic features of the input speech file and ensure noise indistinguishability.

We evaluate \name over a set of real-world datasets covering diverse demographics. Our evaluation shows that \name provides a superior transcription accuracy relative to Deep Speech, the state-of-the-art offline \ASR. Also, \name prevents cloud services from extracting any user-specific acoustic features from the speech. Finally, applying \name thwarts the learning of any statistical models or sensitive information extraction from the text via natural language processing tools.

In summary, the main contributions of this paper are:\\
 \textbf{(1) End-to-end practical system:} We propose \name, a new end-to-end system that provides privacy-preserving speech transcription at an improved performance relative to offline transcription. Specifically, \name shows a relative improvement of  2\% to 32.52\% (mean 17.34\%) in word error rate (\WER) on real-world evaluation datasets over Deep Speech, while fully obfuscating the speakers' voice biometrics and allowing only a \DP view of the textual content. 
\\\textbf{(2) Non-standard use of differential privacy:} \name uses \DP in a \textit{non-standard way}, giving rise to a set of new challenges. {Specifically, the challenges are 
(1) ``noise'' corresponds to concrete words, and need to be added in the speech domain (2) ``noise'' has to be indistinguishable from the original speech (details in Sec.~\ref{sec:DP}).} %
 \\\textbf{(3) Customizable Design:} \name provides several \textit{control knobs} for users to customize the functionality based on their desired levels of utility, usability, and privacy (Sec.~\ref{sec:Q4}). For example, in a relaxed privacy setting, \name's relative improvement in \WER ranges from 44\% to 80\% over Deep Speech (Sec.~\ref{sec:Q4_utility_privacy}).

\ifpaper 
The full version of this paper is available online~\cite{FP}.
\fi  
Some demonstrations of \name are available at this link~\cite{demo}.

\section{Speech Transcription Services}
\label{sec:SpeechTranscription}

We first provide some background on online and offline speech transcription services. Next, we present a utility evaluation using standard and real-world speech datasets. 

\subsection{Background}

Speech transcription refers to the process of extracting text from a speech file. \ASR systems are available to the users either through cloud-based online APIs or offline software. \\
\textbf{(1) Cloud-Based Transcription:} We utilize two cloud-based speech transcription services -- Google's Cloud Speech-to-Text and Amazon Transcribe.
\\
\textbf{(2) Offline Transcription:}
 We consider the Deep Speech architecture from Baidu~\cite{hannun2014deep}, which is trained using Mozilla's \footnote{https://voice.mozilla.org/en/datasets} Common Voice dataset as a representative offline transcription service. This dataset is crowdsourced and open-source.  Specifically, we use the  Deep Speech 0.4.1 model \footnote{https://github.com/mozilla/DeepSpeech} (released in January 2019).  \minrev{Note that we do not consider offline transcribers that are not open for general use. For example, Google's on-device speech recognizer~\cite{google_ondevice} is an offline transcriber that is currently only supported on Google's Pixel devices and does not allow an API or open-source access, limiting its usability. }
 
\paragraph{Notations:} Let $S$ denote the input speech file associated with a ground truth transcript $T^g_S$. The user can either use a cloud service provider (\CSP) or an offline service provider (\OSP) to obtain the transcript (denoted by $T^{CSP}_S$ or $T^{OSP}_S$, respectively). 

\paragraph{Transcription Accuracy:}
\label{sec:wer}
The standard metric for quantifying the accuracy loss from transcription is the word error rate (\WER)~\cite{hannun2014deep}. \WER treats the transcript as a sequence of words. It models the difference between the two sequences by counting the number of deleted words ($D$), the number of substituted words ($U$), and the number of injected words ($I$). If the number of words in $T^g_S$ is $W$, \WER is given as: ${\frac {D+U+I}{W}}$.

\subsection{Utility Comparison}
In this section, we empirically evaluate the utility gap between the \CSP and the \OSP over a wide range of standard and real-world datasets. We use these datasets throughout the paper.

\paragraph{Standard Datasets:}
These datasets include (1) the TIMIT-TEST subset~\cite{garofolo1993darpa}, (2) a subset from Librispeech \textit{dev-clean} dataset~\cite{panayotov2015librispeech}, and (3) the DAPS 
dataset~\cite{mysore2014can}. TIMIT-TEST~\footnote{https://catalog.ldc.upenn.edu/LDC93S1} subset comprises of 1344 utterances by 183 speakers from eight major dialect regions of the United States. The LibriSpeech subset consists of eleven speakers, 20 utterances each. For DAPS, we use the evaluation subset prepared for the 2018 voice conversion challenge~\cite{lorenzo2018voice} that consists of five scripts read by ten speakers: five males and five females.

\paragraph{Real-world Datasets:}
We also assess the real-world performance of both transcription services on non-American accent datasets and real conversations among speakers of different demographics. For the accented datasets, we evaluate 200 utterances of two speakers from the VCTK dataset~\cite{veaux2017cstr}: speaker p262 of a Scottish accent and speaker p266 of an Irish accent. For the real-world datasets, we evaluate 20 minutes of speech from the "Facebook, Social Media Privacy, and the Use and Abuse of Data" hearing before the U.S. Senate~\footnote{https://www.commerce.senate.gov/2018/4/facebook-social-media-privacy-and-the-use-and-abuse-of-data}. We construct the 20 minutes by selecting three continuous chunks of speech from the hearing such that they include nine speakers: 8 senators and Mark Zuckerberg. 
Another real-world dataset is the Supreme Court of the United States case "Carpenter v. United States" ~\footnote{https://www.oyez.org/cases/2017/16-402}. 
For this dataset, we evaluate a total of 40 minutes of speech from the advocates in the case.

\paragraph{Accuracy Comparison:}\label{sec:accuracy_comparison}

Table~\ref{tab:WER_noVC} presents the \WER comparison results. The results show that the \CSPs are superior to the \OSP on all the datasets. The performance gap, however, is more significant on the non-standard datasets; the \CSP outperforms Deep Speech by 60\% to 80\% in \WER. %

\begin{table}[t]
\small
\centering
\ra{1.1}
\scalebox{1}{\begin{tabular}{@{}llcrrr@{}}\toprule
& \textbf{Datasets} & \textbf{Google} & \textbf{AWS} & \textbf{Deep Speech}\\ 
\midrule
\multirow{3}{*}{\rotatebox[origin=c]{90}{Standard}} & LibriSpeech & 9.14 & 8.83 & 9.37\\
& DAPS & 6.70 & 7.53 & 10.65\\
& TIMIT TEST & 6.27 & 7.11 & 20.08\\
\midrule
\multirow{7}{*}{\rotatebox[origin=c]{90}{Non-Standard}} & VCTK p266 & 5.15 & 10.09 & 26.72\\
& VCTK p262 & 4.53 & 7.87 & 15.97\\
& Facebook 1 &5.76 & 7.45 & 24.72\\
& Facebook 2 &3.07 & 8.19 & 26.61\\
& Facebook 3 &8.32 & 9.42 & 30.72\\
& Carpenter 1 & 9.44 &9.44 & 25.85\\
& Carpenter 2 & 9.22 &11.53 & 39.71 \\
\bottomrule
\end{tabular}}%
\caption{\WER(\%) comparison of cloud services, Google and AWS, 
versus the state-of-the-art offline system, Deep Speech.}
\label{tab:WER_noVC}
\end{table} %

\section{Privacy Threat Analysis}
\label{sec:utility_privacy}

We study the privacy threats that a cloud-based transcription service poses while processing private speech data.

\subsection{Voice Analysis}
\label{sec:voice_threat}
The biometric  information embedded in $S$ can leak sensitive information about the speakers, including their emotional status~\cite{schuller2013computational, aloufi2019emotionless}, health condition~\cite{schuller2013paralinguistics}, sex~\cite{safavi2018automatic}, and even identity~\cite{muckenhirn2018towards}. Furthermore, extracting this information enables critical attacks like voice cloning and impersonation attacks~\cite{Wu:2015:SCS:2803221.2803403, lindberg1999vulnerability}. In this section,  we showcase a few representative examples of how cloud-based APIs can pose serious privacy threats to the acoustic features within $S$.

\paragraph{Speaker Diarization:}
\CSPs utilize advanced diarization capabilities to cluster the speakers within a speech file, even if they have not been observed before. The basic idea is to (1) segment the speech file into segments of voice activity, and (2) extract a speaker-specific embedding from each segment, such that (3) segments with close enough embeddings should belong to the same speaker. We verified the strength of the diarization threat over three multi-speaker datasets: VCTK (mixing p266 and p262), Facebook, and Carpenter. We measure the performance of the IBM diarization service using Watson's Speech-to-Text API~\footnote{https://www.ibm.com/cloud/watson-speech-to-text} via Diarization Error Rate (DER). DER estimates the fraction of time the speech file segments are not attributed to the correct speaker cluster. The DER values are 0\%, 4.85\%, and 1.32\% for the three datasets, respectively. Hence, the API can correctly distinguish between, and cluster, the different speakers, more than 95\% of the entire dataset duration despite lacking any prior information about the individual speakers.

\paragraph{Speaker Identification:}
A speaker identification task maps the speech segments in a speech file to an individual. We use the Azure Identification API, which consists of two stages: (1) user enrollment and (2) identification (whether a given voice sample matches any of the enrolled users). The enrollment stage requires only 30 seconds of speech from each user to extract their voice-print.
We enrolled 22 speakers as follows: 10 from DAPS, two from VCTK, two from Carpenter, and eight from Facebook. The identification accuracy was nearly 100\% for all speakers.

\paragraph{Speaker Cloning and Impersonation:}
Lastly, we applied a Tacotron-based speech synthesizer from Google~\cite{google_speech_synthesis}; a network that can synthesize speech in the voice of any speaker. The network generates a target speaker's embedding, which it uses to synthesize speech on a given piece of text. In our setting, we used the network to generate the speakers' embedding in our evaluation datasets.
Then, we synthesized eight speech utterances using the embeddings of each speaker. We enrolled the speakers in Azure's Speech Identification API using their natural voice samples and tested whether the API will map the synthesized segments to the corresponding speaker. Except for the second speaker in Carpenter, the cloned samples were successfully identified as the true speakers. 

\subsection{Text Analysis} 
\label{sec:textAnalysis}

\CSPs possess natural language processing (\NLP) capabilities that enable automated statistical analyses on large sets of documents. Those analyses fall into two broad categories.
The first type involves identifying specific words from the transcript that correspond to sensitive information such as an address, name, and SSN using named-entity extraction~\cite{etzioni2005unsupervised}. The other type of analysis involves statistically analyzing the entire transcript on the whole to extract some semantic or user-identifying information. This analysis uses two types of information: the set of words (i.e., bag-of-words representation of the transcript) and their order of appearance (to capture the context). 
\paragraph{Bag-of-Words Analysis:} One of the most commonplace analysis that treats a document as a bag-of-words is \textit{topic modeling}~\cite{steyvers2007probabilistic, Ramage:2009:LLS:1699510.1699543}. Topic modeling is an unsupervised machine learning technique that identifies clusters of words that best characterize a set of documents.
Another popular technique is \textit{stylometry analysis}, which aims at attributing authorship (in our case, the speaker) of a document based on its literary style. It is based on computing a set of stylistic features like mean word length, words histogram, special character count, and punctuation count from the disputed document~\cite{stylometry}. 
\paragraph{Context-based Analysis:} An example of context-based analysis is sentiment analysis (understanding the overall attitude in a block of text). Text categorization is another example; it refers to classifying a document according to a set of predetermined labels.

\section{\name}
\label{sec:OurSys}

Our discussion in the previous sections highlights a trade-off between privacy and utility. The \OSP provides perfect privacy at the cost of higher error rates, especially for non-standard speech datasets. On the other hand, clear privacy violations accompany revealing the speech recording to the \CSP. Motivated by this trade-off, we present \name, a practical system that lies at an intermediate point along the utility-privacy spectrum of speech transcription.

\subsection{System and Threat Models}
\label{sec:user-system-model}

{ 
We consider the scenario where users have audio recordings of private conversations that require high transcription accuracy. For example, a journalist with recordings of confidential interviews is a paradigmatic user for \name. Other examples include a therapist with recordings of patient therapy sessions or a course instructor with oral examination records of students. \name, however, does not target real-time transcription applications. For example, voice assistants and online transcription (e.g. a live-streaming press conference) are \textit{out-of-scope}. Thus, for our target use cases, the latency of transcription is not a critical concern. }

{The adversary is the \CSP or any other entity having direct or indirect access to the stored speech at the \CSPs. This adversary is capable of the aforementioned voice- and text-based analysis.}

  \begin{figure*}
        \centering
        \includegraphics[width=\textwidth]{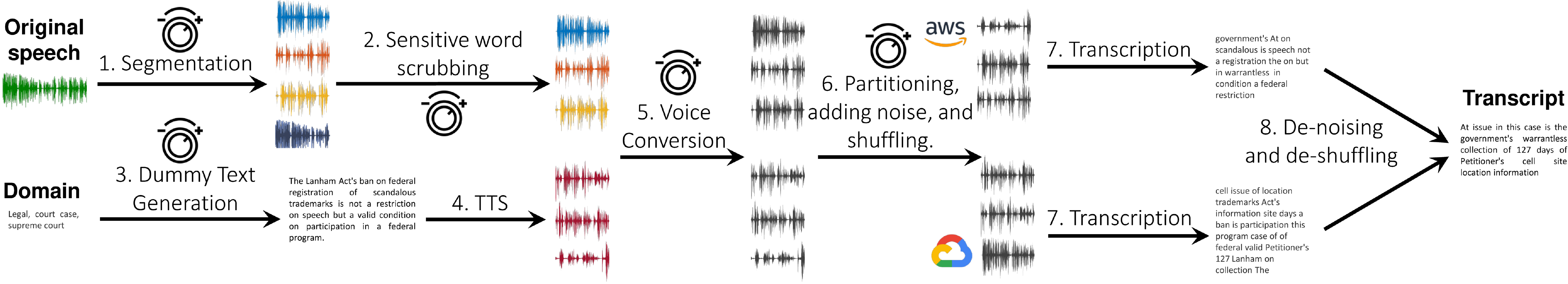}
        \caption{{High-level overview of \name, showing the knobs where a user can tune the associated trade-offs.}}
        \label{fig:high-level}
    \end{figure*}
    
\subsection{\name Overview}

\name provides { an end-to-end tunable system} which aims at satisfying the following design goals: 
 
\begin{enumerate}[leftmargin=*]
\itemsep-0.5em
\item protect the users' privacy along the acoustic and textual dimensions;  
\item improve on the transcription accuracy compared to offline models; and 

\item {provide the users with control knobs to customize \name's functionality according to their desired level of utility, usability, and privacy.
}
\end{enumerate}

\par To this end,  \name applies a series of \textit{privacy-preserving operations} to the input speech file before sending it to the \CSP. 
Fig.~\ref{fig:high-level} shows the high-level overview of \name. Below, we briefly describe \name's privacy-preserving operations.

\subsubsection{Preserving Textual Privacy}
\name protects the privacy of the textual content of an input speech file $S$ through the following three operations:

\paragraph{Segmentation and shuffling:}
\name breaks $S$ into a sequence of segments, denoted by $\mathbb{S}$. This is followed by shuffling the segments to remove all ordering information. Thus, segmenting and shuffling $S$  transform its textual content into a bag-of-words representation.

\paragraph{Sensitive word scrubbing (\SWS):}
First, \name applies the \OSP to identify the list of sensitive keywords that contain numbers, proper nouns, or any other user-specified words.
Next, \name applies keyword spotting, \KWS, (identify portions of the
speech that correspond to a keyword) to each of the
segments in $\mathbb{S}$. Only the segments that do not contain a keyword pass to the \CSP for transcription.

\paragraph{Dummy word injection to ensure differential privacy:}
The bag-of-words representation of a transcript corresponds to its word histogram (Sec.~\ref{sec:DP}). As discussed in Sec.~\ref{sec:textAnalysis}, several statistical analyses can be built on the word histogram of the transcript $T^{CSP}_S$ such as topic modeling or stylometry analysis.
Thus, protecting the privacy of this word histogram is a primary focus of \name, and the privacy guarantee we choose is that of differential privacy. To this end, \name ensures \DP  by adding a suitable amount of dummy words to $S$ before sending it to the \CSP. This way, the \CSP is allowed only a differentially private view of the word histogram 
and any subsequent statistical model built over it (by Thm. \ref{thm:post} in Sec.~\ref{sec:DP}).

The main challenge in this setting is that the dummy words must be added in the speech domain, which \name addresses as follows.  First,  \name estimates the general domain of the text for $S$ (specifically its \textit{vocabulary}, details in Sec.~\ref{sec:DP}) from $T^{OSP}_S$. {Next, it generates dummy text segments using a state-of-the-art \NLP language model.} 
Finally, \name applies text-to-speech (\TTS) transforms to these dummy segments and adds them to $S$. However, leaving it just at this would be insufficient as the \CSP can potentially distinguish between the two different sources of speech (\TTS generated dummy segments and segments in $S$) based on their acoustic features. {Therefore, \name provides the user with multiple options to synthesize \textit{indistinguishable} dummy segments, namely (1) voice cloning~\cite{google_speech_synthesis}, and (2) voice conversion~\cite{kobayashi2018sprocket, VC}. These options offer different trade-offs between utility, usability, and privacy (Secs. \ref{sec:DP:discussion} and \ref{sec:VoiceConversion}}). As stated in Sec.~\ref{sec:textAnalysis}, text-based attacks exploit individual sensitive words or the order of the words or the word histogram. Thus, from the above discussion, \name protects privacy along all three dimensions (evaluation results in Sec.~\ref{sec:evaluation}). 

\subsubsection{Preserving Voice Privacy}
Voice conversion, \VC, is a standard speech processing technique that transforms the voice of a source speaker of a speech utterance to that of another speaker. \name applies voice conversion to fulfill a two-fold agenda. First, it obfuscates the sensitive voice biometric features in $S$. Second, \VC ensures that the dummy segments (noise added to ensure differential privacy) are acoustically indistinguishable from the original speech file segments. There are two main categories in voice conversion: one-to-one \VC, and many-to-one \VC ( Sec.~\ref{sec:VoiceConversion}).

\subsubsection{End-to-End System Description} \label{sec:end-to-end}
Fig.~\ref{fig:high-level} depicts the workflow of \name. Given a speech file $S$, the first step \textbf{(1)} is to break $S$ into a sequence of disjoint and short speech segments, $\mathbb{S}$. This is followed by \textbf{(2)} sensitive word scrubbing where speech segments containing numbers, proper nouns, and user-specified keywords are removed from $\mathbb{S}$.  Next, \textbf{(3)} given the domain of $S$'s textual content (its vocabulary), \name generates a set of text segments (as is suitable for satisfying the \DP guarantee as discussed in Sec.~\ref{sec:DP}), and subjects it to \TTS transformation \textbf{(4)}. At this point, \name has audio segments for the input speech, $\mathbb{S}$, as well as the dummy segments, $\mathbb{S}_d$. If the user also wants to hide the voice biometric information in $S$, \name applies \textbf{(5)} voice conversion over all the segments in $\mathbb{S} \bigcup \mathbb{S}_d$ to convert them to the same target speaker. This process hides the acoustic features of $S$ and ensures that the segments in $\mathbb{S}$  and   $\mathbb{S}_d$ are indistinguishable. This is followed by \name partitioning $\mathbb{S}$ across $N>0$ non-colluding \CSPs (Sec.~\ref{sec:DP}). This partitioning reduces the number of dummy segments that are required to achieve the \DP guarantee (Sec.~\ref{sec:DP}). Next, \name adds a suitable amount of dummy segments from $\mathbb{S}_d$ to each partition $\mathbb{S}_i, i \in [N]$  and shuffles them. Additionally, \name keeps track of time-stamps of the dummy segments, $TS_i$ and order of shuffling, $Order_i$ for each such partition \textbf{(6)}. After obtaining the transcript \textbf{(7)} for each partition from the $N$ \CSPs, \name removes
$\mathbb{S}_d$'s transcripts and de-shuffles the remaining portion of the transcript using $TS_i$ and $Order_i$, and outputs the final transcript to the user \textbf{(8)}. 

\noindent
In what follows, we elaborate on the key components of \name, namely segmentation, sensitive word scrubbing, \DP word histogram release, and voice conversion.

\subsection{Segmentation Algorithm}
\label{sec:Segmentation}
  \begin{figure}
        \centering
        \includegraphics[width=0.9\columnwidth]{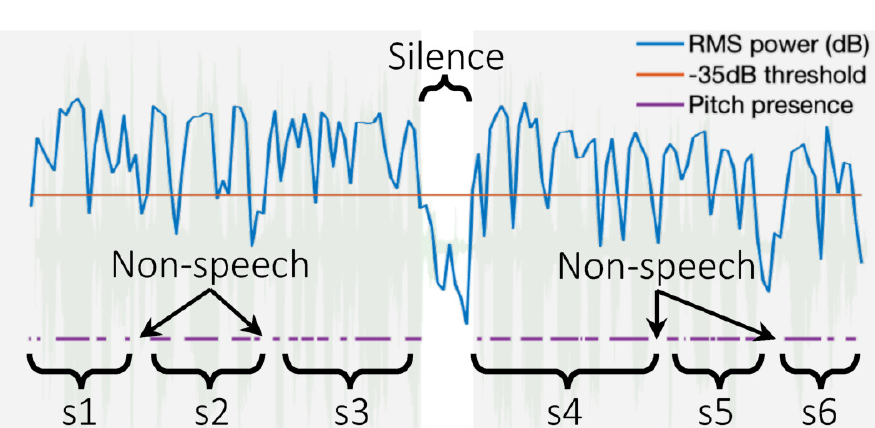}
        \caption{An illustration of \name's segmentation algorithm. The coarse segments in light gray. The absence of pitch information indicate non-speech instances, which further breaks down the coarse segments into finer segments.}
        \label{fig:segmentation}%
    \end{figure}

A key component of \name is breaking the textual context by segmenting $S$. We represent $S$ as a sequence of segments $\mathbb{S}$, where each segment can contain multiple words. \name applies a hierarchical segmentation approach that starts with a stage of silence detection based on the energy level, followed by pitch detection to detect speech activity for finer segmentation. The mechanism is illustrated in 
Fig.~\ref{fig:segmentation}.

We define a \textit{period of silence} as the time duration when the RMS power of the speech signal drops below -35 dB for at least 500ms. The initial segmentation stage detects such silence periods from $S$ resulting in coarse segments. A human speech signal can be viewed as a modulated periodic signal where the signal period 
 is referred to as the \textit{glottal cycle}~\cite{murty2009characterization}. 
 In the second stage, \name uses the existence of glottal cycles~\cite{Boersma:1993} to detect human voice, which breaks down the coarse segments into finer ones. A time duration of at least 20 ms without the presence of glottal cycles is regarded as \textit{non-speech}. 
 
 As some segments might be abrupt or too short to allow for correct speech recognition, \name performs two additional optimization steps. First, it merges nearby fine segments to ensure a minimum length per segment. Second, it does not partition segments at the boundaries of the identified human speech and allows 40 ms of non-speech to be included at the beginning and the end of each segment. \\

\noindent
{\paragraph{Control Knob:} Segmenting $S$ presents with a trade-off -- smaller segments result in better privacy guarantee at the expense of deteriorated transcription accuracy due to semantic context loss. \name allows the user to tune the \textit{minimum length of the segments as a means to control this trade-off.}} \par 

\subsection{Sensitive Word Scrubbing}
\label{sec:sws}
\name performs sensitive word scrubbing (\SWS) as follows. First, it obtains the offline transcript of $S$, $T^{OSP}_S$. Next, it applies named entity recognition (\NER) on $T^{OSP}_S$.
\NER is an  NLP technique that seeks to locate and classify named entities in text into pre-defined categories such as the names of persons, organizations, locations, expressions of times, monetary values, etc.  \name also gives the option for users to specify some keywords of their choice. This allows customization of the sensitive keyword list as users have subjective ideas of what they might consider sensitive. 

After the list of sensitive words is finalized, \name applies keyword spotting (\KWS) on the segments. 
\KWS  is needed for the following three reasons. First, \KWS is used to spot
the user-defined keywords which cannot be identified by \NER. Second, the initial $T^{\OSP}_S$ is generated on $S$ without segmentation to achieve the highest estimation accuracy. However, for \name, we need to identify the segments containing the keywords. Finally, the \OSP
 might not transcribe the named-entities correctly at all locations. For example, the name ``Carpenter'' might be repeated 20 times in $S$, while the \OSP transcribes it accurately only five times. \KWS has higher accuracy in spotting keywords than the \OSP's transcription accuracy.

\paragraph{Control Knob:}  \KWS takes the list of keywords and matches them phonetically to a speech file based on a sensitivity score. This sensitivity score sets a threshold for the phonetic similarity required for a keyword to be spotted.  A low score results in false positives by flagging phonetically similar words as keywords which degrades the utility by transcribing non-sensitive segments using the \OSP. Conversely, a high score could result in some keywords being missed and revealed to the \CSP. Hence, the sensitivity score is a trade-off parameter between privacy and utility (Sec.~\ref{sec:eval:sws}).

\subsection{Differentially Private Word Histogram}\label{sec:DP}

We define $vocabulary, \mathcal{V}$, to be the domain of non-stop and stemmed words from which $T^g_S$ is constructed. 
Let $c_i$ denote the frequency of the word $w_i \in \mathcal{V}$ in $T^g_S$. As is typical in the \NLP literature, we model the transcription as a bag of words: $BoW = \{w_i:c_i|w_i \in \mathcal{V}\}$. Additionally, let $H$ represent $[c_i]$ -- the count vector of $BoW$. %
In other words, the bag of words model represents a histogram on the vocabulary, i.e., a mapping from $\mathcal{V}$ to $ \mathbb{N}^{|\mathcal{V}|}$.

\subsubsection{Privacy Definition}

As discussed in Sec. \ref{sec:textAnalysis}, the aforementioned word histogram is sensitive and can only be released to the \CSP in a privacy-preserving manner.  Our privacy guarantee of choice is \DP which is the de-facto standard for achieving data privacy \cite{dwork2014algorithmic, chen2011publishing, friedman2010data}. \DP provides provable privacy guarantees and is typically achieved by adding noise to the sensitive data.  
\begin{definition}[$(\epsilon,\delta)$-differentially private $d$-distant  histogram release] A randomized mechanism $\mathcal{A}: \mathbb{N}^{|\mathcal{V}|}
\rightarrow \mathbb{N}^{|\mathcal{V}|}$, which maps the original histogram into a noisy one, satisfies $(\epsilon, \delta)$-DP if for any pair of histograms $H_1$ and $H_2$ such that $||H_1-H_2||_1 = d$ and any set $O \subseteq \mathbb{N}^{|\mathcal{V}|}$,
\begin{gather}  Pr[\mathcal{A}(H_1)\in O]\leq  e^{\epsilon}\cdot Pr[\mathcal{A}(H_2)\in O] + \delta.\end{gather}\label{def:DP}\vspace{-0.5cm}
\end{definition}
In our context, the DP guarantee informally means that from the $CSP$'s perspective, the observed noisy histogram, $\tilde{H}$, could have been generated from any histogram within a distance $d$ from the original histogram, $H$. We define the set of all such histograms to be the $\epsilon$-indistinguishability neighborhood for $H$. In other words, from $\tilde{H}$ the $CSP$ will not be able to distinguish between $T^{CSP}_S$ and any other transcript that differs from $T^{CSP}_S$ in $d$ words from  $\mathcal{V}$.  \par

An important result for differential privacy is that any post-processing computation performed on the output of a differentially private algorithm does not cause any loss in privacy.
\begin{theorem}(Post-Processing)\label{thm:post}
Let $\mathcal{A}:  \mathcal{X} \mapsto R$ be a randomized
algorithm that is $(\epsilon,\delta)$-DP. Let $f : R \mapsto R'$ be an
arbitrary randomized mapping. Then $f \circ \mathcal{A} : \mathcal{X} \mapsto R'$ is $(\epsilon,\delta)$-
DP. \end{theorem}
Another result is that the privacy of DP-mechanism can be amplified if it is preceded by a sampling step.
\begin{theorem} Let $\mathcal{A}$ be an $(\epsilon,\delta)$-DP algorithm and $\mathcal{D}$ is an input dataset. Let $\mathcal{A}'$ be another algorithm that runs $\mathcal{A}$ on a random subset of  $\mathcal{D}$ obtained by sampling it with probability $\beta$.  Algorithm $\mathcal{A}'$ will satisfy $(\epsilon',\delta')$-DP where $\epsilon'=ln(1+\beta(e^\epsilon -1))$ and $\delta'<\beta\delta$. \label{thm:sampling}
\end{theorem}
Additionally, we define a DP mechanism namely the truncated Laplace mechanism \cite{shrinkwrap} which is used in \name.
\begin{definition}[Truncated Laplace mechanism for histogram] Given a histogram $H$, the truncated Laplace mechanism, $Lp(\epsilon,\delta,d)$, adds a non-negative integer noise vector $[\max(\eta,0)]^{|\mathcal{V}|}$ to $H$, where $\eta $ follows a distribution, denoted by $L(\epsilon,\delta,d)$ with a p.d.f  $\Pr[\eta=x] = p \cdot e^{- (\epsilon/d)|x-\eta^0|}$, where $p = \frac{e^{\epsilon/d} - 1}{e^{\epsilon/d}+ 1}$ and $\eta_0 = -\frac{d \cdot \ln((e^{\epsilon/d} + 1)\delta)}{\epsilon} + d$.\end{definition}
\begin{theorem}The truncated Laplace mechanism satisfies $(\epsilon, \delta)$-DP for $d$-distant histogram releases~\cite{shrinkwrap}. \label{thm:truncated}\end{theorem}

\begin{figure*}[t]     
\centering
\begin{subfigure}{0.24\linewidth}
\vspace{0.1cm}
  \centering
  \includegraphics[width=\linewidth]{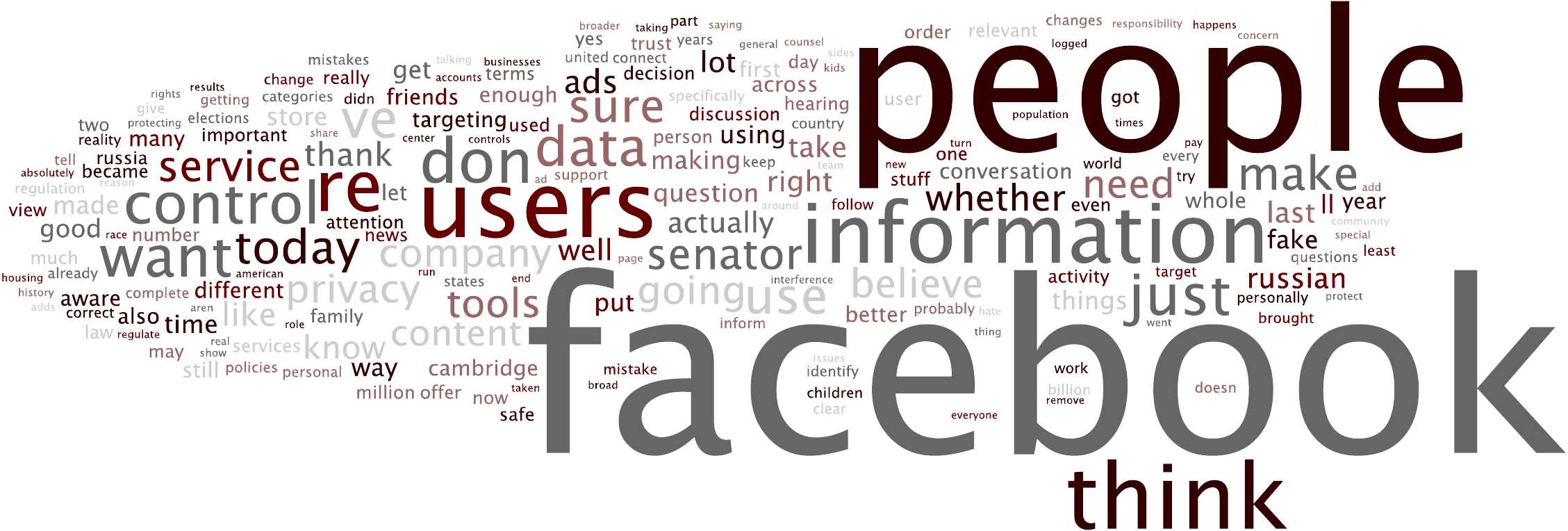}
  \vspace{0.1cm}
  \caption{Original}
  \label{fig:cloud_1} 
\end{subfigure}
\begin{subfigure}{0.24\linewidth}
  \centering
  \includegraphics[width=\linewidth]{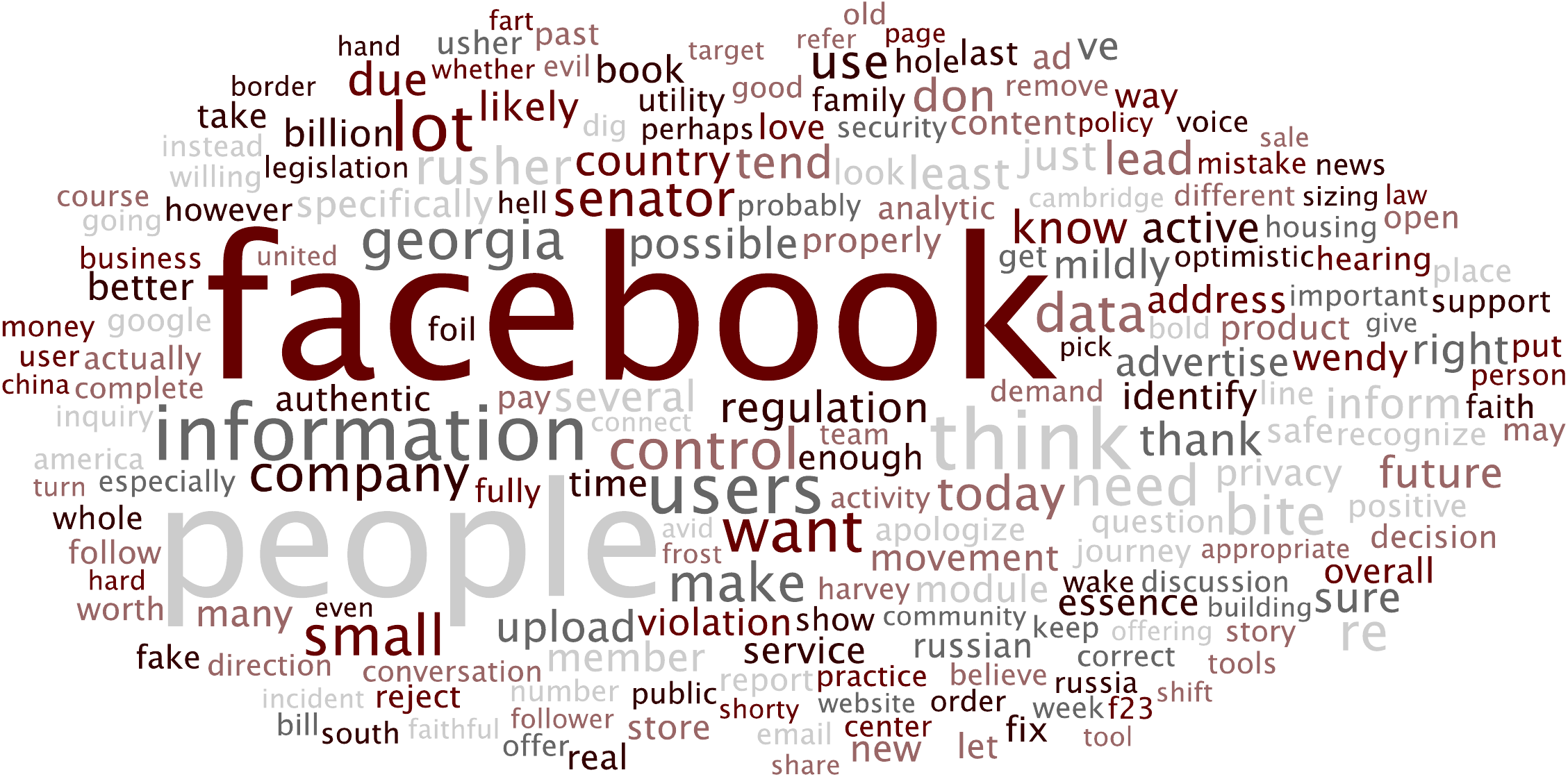}
  \caption{$\epsilon$=1, $\delta$=0.05, and $d$=2}
  \label{fig:cloud_2} 
\end{subfigure}
\begin{subfigure}{0.24\linewidth}
  \centering
  \includegraphics[width=\linewidth]{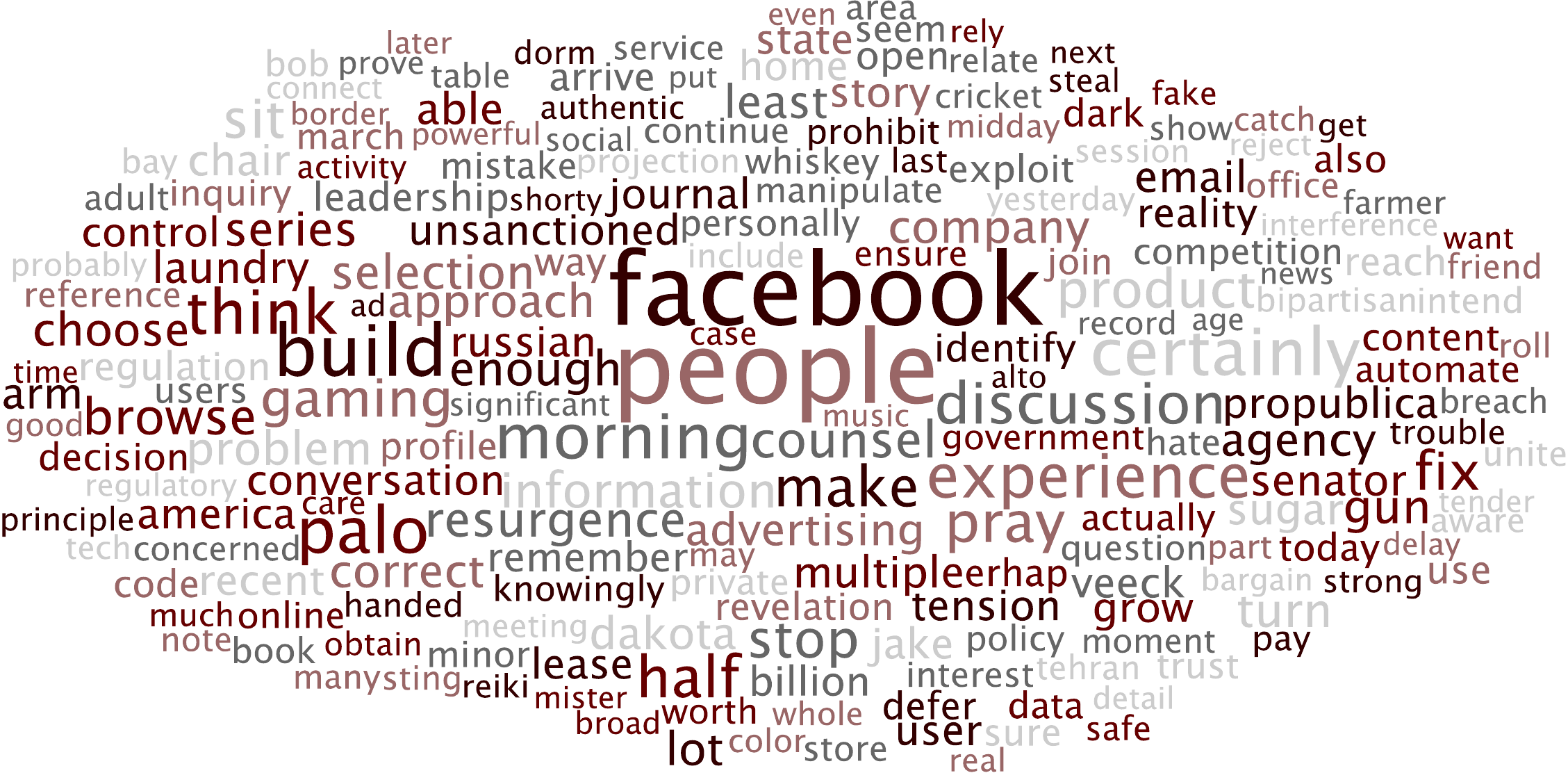}
  \caption{$\epsilon$=1, $\delta$=0.05, and $d$=5}
  \label{fig:cloud_3} 
\end{subfigure}
\begin{subfigure}{0.24\linewidth}
  \centering
  \includegraphics[width=\linewidth]{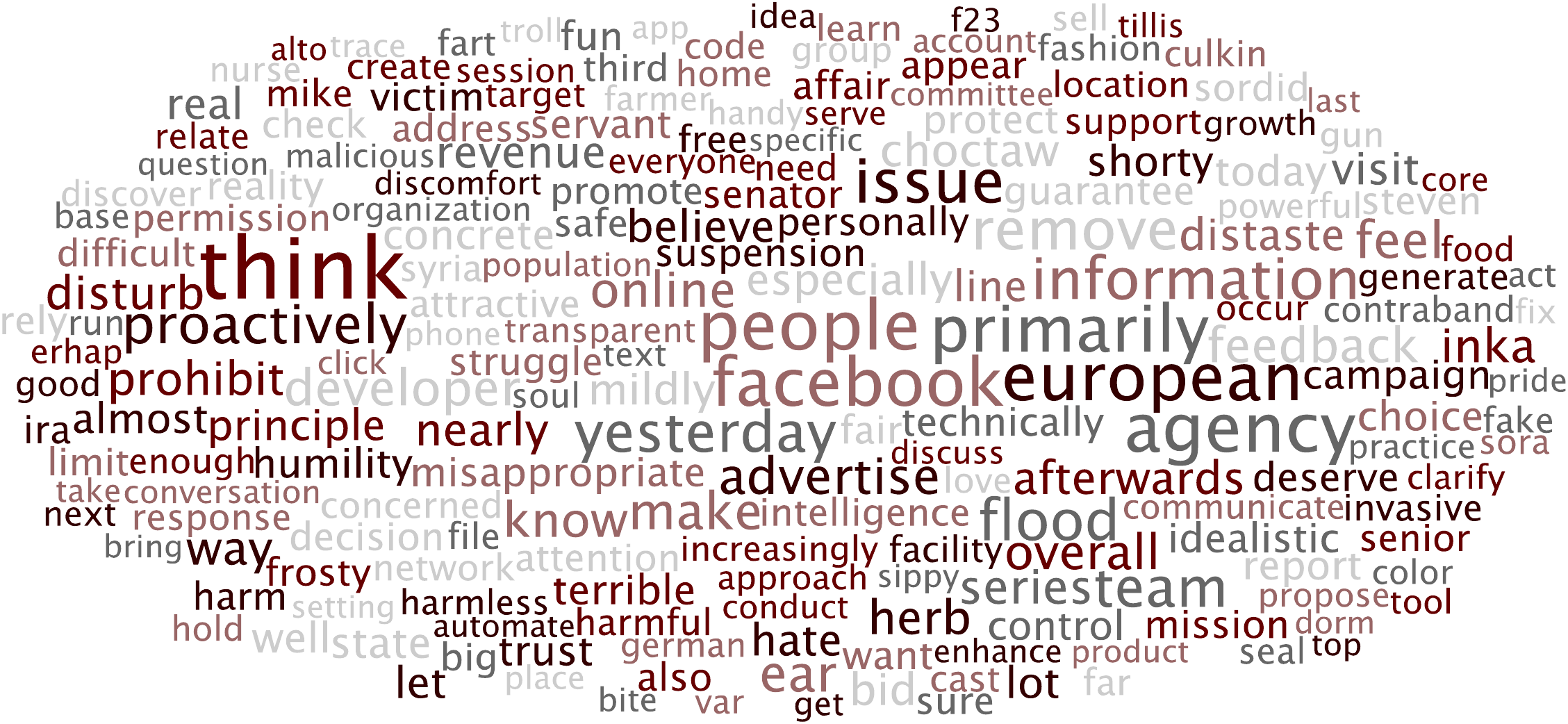}
  \caption{$\epsilon$=1, $\delta$=0.05, and $d$=10}
  \label{fig:cloud_4} 
\end{subfigure}
\caption{The word cloud of the Facebook dataset visualizing the histogram as it changes after adding different levels of noise.}
\label{fig:cloud_privacy}
\end{figure*}

Fig.~\ref{fig:cloud_privacy} visualizes the histogram of the Facebook dataset as a word cloud for different noise levels. As evident from the original word cloud, the histogram emphasizes few important words such as {\fontfamily{cmss}\selectfont Facebook}, {\fontfamily{cmss}\selectfont people}, {\fontfamily{cmss}\selectfont information}, and {\fontfamily{cmss}\selectfont users}. With increased value of $d$, the resulting histogram has a roughly uniform distribution of the included words.

\subsubsection{Discussion}\label{sec:DP:discussion}
\name's use of DP is different from the most standard use-case of DP (like numeric datasets). It deals with concrete units like words instead of numeric statistics -- introducing new challenges; we discuss these challenges and how \name circumvents them in this section.\\
\textbf{Vocabulary definition: }
 The foremost task for defining the word histogram is defining the vocabulary, $\mathcal{V}$. The most conservative approach to define $\mathcal{V}$ is to consider the total set of all English stemmed and non-stop words. Such a vocabulary would be prohibitively large for efficient and practical usage. However, note that such a definition of $\mathcal{V}$ is an overestimate as no real-world document would contain all possible English words. Recall that our objective of adding noise is to obfuscate any statistical analysis built on top of the document's $BoW$ (histogram), such as a topic modeling and stylometry analysis. Typically, $BoW$ based statistical analyses are concerned only with the set of most frequent words. For example, any standard topic model captures only the top  $m$ percentile most frequent words in a transcript \cite{steyvers2007probabilistic, Ramage:2009:LLS:1699510.1699543}. The same applies to stylometry analysis, which is based on measures of the unique distribution of frequently used words of different individuals.

Thus, as long as the counts of the most common words of the transcript are protected (via DP), the subsequent statistical model (like topic model) built over the word histogram will be privacy-preserving too (by Thm. \ref{thm:post}). {However, high-frequency words might not be the only ones that contain important information about $T_S$. To tackle this, we also include words with large Term Frequency-Inverse Document Frequency (\TFIDF) weight to our vocabulary. This weight is a statistical measure used to evaluate how significant a word is to a document relative to a baseline corpus. The weight increases proportionally to the number of times a word appears in the document but is offset by the frequency of the word in the baseline corpus. This offset adjusts for the fact that some words appear more frequently in general.  To this end, \name makes an estimate of the vocabulary from $T^{OSP}_S$. Although existing offline transcribers have high \WER, we found (empirically) that they can identify the set of domain words of $S$ with high accuracy (details in Sec. \ref{sec:eval:text}). For computing the TF-IDF values, IDF is computed using an external NLP corpus like Wikipedia articles. Thus formally, $\mathcal{V}=\{w| w \in \{$ top $m$ percentile of the most frequent words in $T^{OSP}_S\}\cup \{$ words with TF-IDF value $\geq \Delta \mbox{ in } T^{OSP}_S\}\}$. 
Note that $\mathcal{V}$ should be devoid of all sensitive words which are scrubbed off from $S$ in step 2 of Fig. \ref{fig:high-level}. Additionally, the vocabulary can be extended to contain \textit{out-of-domain} words, i.e., random English words that are not necessarily part of the original document. This helps in protecting against text classification attacks (Sec.~\ref{sec:eval:text}).} 

\paragraph{Specificities of the word histogram:} As discussed above, the goal of the DP mechanism is to generate noisy counts for each $w_i \in \mathcal{V}$. An artifact of our setting is that this noise has to be non-negative and integral. This is because dummy words (for the noisy counts) can only be added to $S$; removing any word from $S$ is not feasible as this would entail in recognizing the word directly from $S$, which would require accurate transcription. Hence, \name uses the truncated Laplace mechanism to ensure non-negative and integral noise. 

\paragraph{Setting privacy parameters:} 
The parameters $\epsilon$ and $\delta$ quantify the privacy provided by a DP-mechanism; lower the values higher is the privacy guarantee achieved. {The distance parameter $d$, intuitively, connects the privacy definition in the word histogram, which is purely a formal representation, to a semantic privacy notion. For example, it can quantify how much the noisy topic models computed by the \CSP (from $T^{CSP}_S$) should differ from that of $T^g_S$. Thus, the user can tune $d$ depending on the target statistical analysis.  In the following, we detail a mechanism, as a guide for the user, for choosing $d$ when the target statistical analysis is topic modeling. }
\par Let us assume that the user has a set of speech files $\{S_j\}$ to be transcribed. Let $D_j$ denote the ground truth transcript corresponding to speech file $S_j$. The objective is to learn $t$ topics from the corpus $\bigcup_j D_j$ with at least $k$ words per topic (a topic is a distribution over a subset of words from the corpus). Let $\mathcal{T}=\{\mathbb{T}_1,\cdots,\mathbb{T}_t\}$ represent the original topic model built on $\bigcup_j D_j=\bigcup_jT^g_{S_j}$  and $\mathcal{T}'=\langle \mathbb{T}'_1, \cdots, \mathbb{T}'_t\rangle$ represent the noisy topic model computed by the \CSP.

The following theorem (Thm.~\ref{thm:topic_bound}) provides a lower bound on the pairwise $\ell_1$ distance between the true and noisy topics as a function of the privacy parameters of the DP word histogram release mechanism (specifically, the term $C_{min}$ is a function of $(d, \epsilon, \delta)$).

\begin{theorem}
\label{thm:topic_bound}
For any pair of topics $(\mathbb{T},\mathbb{T}') \in \mathcal{T}\times \mathcal{T}'$, 
\begin{gather*}
{\textstyle
||T-T'||_1 \geq 2  \frac{1}{\Big(1-(t-1) \frac{k}{\max_j|D_j|}\Big)}  \Big(\frac{\mathcal{C}_{min}}{t} -\frac{1}{2} \Big(1-t \frac{k}{\max_j|D_j||}\Big)\Big),} 
\end{gather*} 

where $\mathcal{C}_{min}=min_{j,l}\Big\{\frac{v\cdot(|D_j|-|w_{l,j}|\omega_j)}{|D_j|\cdot(|D_j|+v\cdot \omega_j})\Big\}$, $|D_j|$  is the total number of words in $D_j$, $\omega_j$ is the total number of  unique words, $v$ is the variance of the distribution 
$Lp(\epsilon', \delta',d),$
$\delta'=\beta\delta$ and $|w_{l,j}|$ is the number of times the word 
 $w_l \in \mathcal{V}$ appears in  $D_j$.\label{thm:DP:topicModeling}
\end{theorem}
The proof of this theorem and the descriptions of the parameters are presented in \ifpaper the full paper~\cite{FP} \else Appendix A\fi.  
\paragraph{Dummy word injection:}  As discussed earlier, achieving differential privacy requires adding dummy words to $S$. {\name generates the dummy text corpus using an NLP language model (Sec. \ref{sec:implementaton}). The model takes in a short text sample from the required topic and generates an entire document of any required length based on that input. In some scenarios, the user can also provide a corpus of non-publicly available documents with the same vocabulary. This scenario is valid in many practical settings. For instance, in an educational institution, the sensitive speech files requiring transcription might be the interviews/oral exams of the students conducted on a specific subject, and the noise corpus can be the lecture notes of the same subject. }

Next, \name generates a set of dummy segments, $\mathbb{S}_d$, from the dummy corpus above. Let us assume that each of the true segments contains at most $k$ non-stop words (depends on the segment length). 
\name ensures that each dummy segment also contains no more than $k$ non-stop words. Additionally, each such segment must contain only one word from the vocabulary $\mathcal{V}$. This means that although the physical noise addition is carried at the segment level, it is still equivalent to adding noise at the level of words (belonging to $\mathcal{V}$) as we only care about $w_i \in \mathcal{V}$. Each dummy segment is injected only once per \CSP.   
Since the dummy segments have to be added in the speech domain, \name applies \TTS transforms to the segments in $\mathbb{S}_d$ such that they have the same acoustic features as $\mathbb{S}$. This condition ensures that $\mathbb{S}_d$ are indistinguishable from $\mathbb{S}$ in terms of their acoustic features. {\name provides the user with two broad options to satisfy this condition --  voice cloning or voice conversion.}

\par Voice cloning is a \TTS system that generates speech in a target speaker voice. Given a speech sample from the target speaker, the system generates an embedding of the speaker's voice biometric features. It uses this embedding to synthesize new utterances of any linguistic content in the target speaker's voice. \name utilizes such a technology to clone the original speaker's voice and uses it to generate acoustically similar dummy segments $\mathbb{S}_d$. \name applies a state-of-the-art voice cloning system~\cite{google_speech_synthesis}, which generates a close-to-natural synthetic voice using a short ($\sim$ 5 sec.) target voice sample.

We evaluate this cloning system in Sec.~\ref{sec:voice_threat}, and the cloned samples are successfully identified as the true speakers. However, voice cloning does not protect the speakers' voice biometrics, \minrev{and can be potentially thwarted by a stronger adversary.} Hence, \name provides voice conversion (\VC) as a stronger privacy-preserving option for the user. \VC transforms the voice of a source speaker to sound like a target speaker. \name utilizes \VC to obfuscate the true speakers' voice biometrics as well as to mitigate the \DP noise indistinguishability concern by converting the true and dummy segments into a single target speaker voice (Sec. \ref{sec:VoiceConversion}). \minrev{We discuss the utility-privacy trade-offs of both options in Sec.~\ref{sec:evaluation}.}

\par It is important to note that the dummy segments do not affect the \WER of $T^{CSP}_S$.  It is so because \name can exactly identify all such dummy segments (from their timestamps) and remove them from $T^{CSP}_S$. Additionally, since the transcription is done one segment at a time, the dummy segments do not affect the accuracy of the true segments ($\mathbb{S}$) either. Segmentation and voice conversion are the culprits behind the \WER degradation, as will be evident in Sec. \ref{sec:evaluation}.  Thus in \name,  the noise (in the form of dummy segments) can ensure differential privacy without affecting the utility. This is in contrast to standard usage of differential privacy for releasing numeric statistics where the noisy statistics result in a clear loss of accuracy. However, the addition of the dummy segments in \name does increase the monetary cost of using the online service that has to transcribe more speech data than needed. We analyze this additional cost in Sec.~\ref{sec:evaluation}. 
\par In practice, we have multiple well-known cloud-based transcription services with low \WER like Google Cloud Speech-to-Text, Amazon Transcribe, etc. \name uses them to its advantage in the following way.  \name splits the set of segments $\mathbb{S}$ into $N$ different sets (step 3 in Sec.~\ref{sec:DPmechanism}) $\mathbb{S}_i, i \in [N]$ where $N$ is the number of \CSPs with low \WER. Then, \name sends each subset to a different \CSP (after adding suitable noise segments to each set and shuffling them). Since each engine is owned by a different, often competing corporation, it is reasonable to assume that the \CSPs are \textit{non-colluding}.   
Thus, assuming that each segment contains at most one word in $\mathcal{V}$,  each subset of segments $\mathbb{S}_i$ can be viewed as randomly sampled sets from $\mathbb{S}$ with sampling probability $\beta =1/N$. From Thm.~\ref{thm:sampling}, this partitioning results in a privacy amplification. 

\subsubsection{Mechanism}
\label{sec:DPmechanism}

We summarize the \DP mechanism by which \name generates the dummy segments for $S$.  The inputs for the mechanism are (1) $\mathbb{S}$ -- the short segments of the speech file $S$, (2) the privacy parameters $\epsilon$ and $\delta$ and (3) $N$ -- the number of non-colluding \CSPs  to use. This mechanism works as follows:
\squishlist
\itemsep-0.3em 
    \item Identify the vocabulary $\mathcal{V}=\{w| w \in \{$ top $m$ percentile of the most frequent words in $T^{OSP}_S\}\cup \{$ words with TF-IDF value $\geq \Delta \mbox{ in } T^{OSP}_S\}\}$ through running an offline transcriber over $S$.
    \item Tune the value of $d$ based on the lower bound from Thm.~\ref{thm:topic_bound}, $\epsilon$ and $\delta$.
    \item Generate $N$ separate noise vectors, $\eta_i \sim [Lp((\ln(1 + \frac{1}{\beta} (e^\epsilon-1)), \beta\delta,d)]^{|\mathcal{V}|}, i \in [N]$. Thus for every partition $i$, \name associates each word in $\mathcal{V}$ with a noise value, a non-negative integer.
    \item From the \NLP generated text, extract all the text segments that contain words from $\mathcal{V}$. For each partition $i$, sample the text segments from this corpus to match the noise vector $\eta_i$. This is the set of noise (dummy) segments for partition $i$, $\mathbb{S}_{d,i}$. Iterate on generating text from the \NLP language model until the required noise count is satisfied. 
    \item Randomly partition $\mathbb{S}$ into $N$ sets $\mathbb{S}_i, i \in [N]$ where $\mbox{Pr[segment $s$ goes to partition $i$}]= \beta=1/N, s \in \mathbb{S}$.
    \item For each partition $i \in [N]$, shuffle the dummy segments in $\mathbb{S}_{d,i}$ (after applying \TTS and \VC) with  the segments in $\mathbb{S}_i$ (after applying \VC), and send it to the $\textsf{CSP}_i$.
\squishend

The first 4 steps in the above mechanism are performed in stage 3 in \name (Fig. \ref{fig:high-level}) while steps 5-6 are performed in stage 6.
\begin{theorem}Any topic model computed by $CSP_i, i \in [N]$ from $T^{CSP_i}_S$ is $(\epsilon,\delta)$-DP. \end{theorem}
\begin{proof}From Thm.~\ref{thm:sampling} and Thm.~\ref{thm:truncated}, we conclude that the word histogram $\tilde{H}_i$ computed from $T^{CSP_i}_S$ is $(\epsilon,\delta)$ - DP for distance $d$. Thm.~\ref{thm:post} proves that the topic model from $\tilde{H}_i$  is still $(\epsilon,\delta)$-DP as it is a post-processing computation. \end{proof}

{\subsubsection{Novelty of \name's Use of Differential Privacy }
\label{sec:DPnovelty}
Here, we summarize the key novelty in \name's use of \DP:}\\
(1) {Typically, DP is applied to statistical analysis of numerical data where "noise" corresponds to numeric values. In contrast, in \name, "noise" corresponds to concrete units -- words. To tackle this challenge, we applied a series of operations  (segmentation, shuffling, and partitioning) to transform the speech transcription into a $BoW$ model, where the \DP guarantee can be achieved. Moreover, the noise addition has to be done in the speech domain. This constraint results in new challenges: the lack of a priori access to the word histogram domain $\mathcal{V}$, and generating indistinguishable dummy speech segments.}\\
(2) {In our setting, the use of a \DP mechanism does not introduce a privacy-utility trade-off from the speech transcription standpoint.  \minrev{\name performs transcription one segment at a time. It keeps track of the timestamps of the dummy segments and completely removes their corresponding text from the final transcription (Sec. \ref{sec:end-to-end}). This filtration step is achievable in \name, unlike numeric applications of \DP, because of the atomic nature of transcription.} However, the dummy segments increase the monetary cost of transcription, resulting in a privacy-monetary cost trade-off as shown in Table~\ref{tab:noise_cost}. To tackle this issue, \name takes advantage of the presence of multiple \CSPs (Sec. \ref{sec:DP:discussion}). Thus, the idea of utilizing multiple \CSPs for cost reduction (Thm. \ref{thm:sampling}) is a novel contribution.}\\
(3) We introduce an additional parameter $d$, the distance between the pair of histograms, in our privacy definition (Defn.~\ref{def:DP}). Intuitively,  $d$ connects the privacy definition in the word histogram model, which is purely a formal representation, to a semantic privacy notion (e.g., $\ell_1$ distance between true and noisy topic models, Thm.~\ref{thm:DP:topicModeling}) as shown in Fig.~\ref{fig:heatmap} and \ref{fig:topic_distance}. \minrev{This contribution builds on ideas like group privacy \cite{dwork2014algorithmic} and generalized distance metrics \cite{generaliseddist}.} 

\vspace{-0.2cm}
{\subsubsection{Control Knobs} \label{sec:DP:control_knobs}
The construction of the DP word histogram provides the user with multiple control knobs for customization:}

\noindent
{\textbf{Parameter $d$:} According to Def. \ref{def:DP}, from $\tilde{H}$ the $CSP$ will not be able to distinguish between $T^{CSP}_S$ and any other transcript that differs from $T^{CSP}_S$ in $d$ words from  $\mathcal{V}$. Thus, higher the value of $d$, larger is the $\epsilon$-indistinguishability neighborhood for $\tilde{H}$ and hence, better is the privacy guarantee. But it results in an increased amount of noise injection (hence, increased monetary cost -- details in Sec.~\ref{sec:eval:text}).}

\noindent
{\textbf{Vocabulary:} The size of $\mathcal{V}$ is a control knob, specifically, the parameters $m$ and $\Delta$ and the number of \textit{out-of-domain} words. The trade-off here is: the larger the size of $\mathcal{V}$, the greater is the scope of the privacy guarantee. However, the noise size scales with $|\mathcal{V}|$ and hence incurs higher cost (details in Sec. ~\ref{sec:eval:text}).}

\noindent

\noindent
{\textbf{Voice transformation for noisy segments:} 
\name provides two options for noise synthesis -- voice cloning and voice conversion. { Voice cloning does not affect the transcription utility, measured in \WER, because it does not apply any transformations on the original speaker's voice}. However, it fails to protect the sensitive biometric information in $S$. \minrev{Moreover, there is no guarantee that a strong adversary cannot develop a system that can distinguish the cloned speech segments from the original ones. This puts \name's effectiveness at the risk of the arms race between the voice cloning system's performance and the adversary's strength.} This limitation is addressed by voice conversion at the cost of transcription utility. We quantify these utility-privacy trade-offs in Sec.~\ref{sec:evaluation}. }

\noindent
{\textbf{Number of \CSPs used for transcription:} As discussed above, employing multiple \CSPs lowers the monetary cost incurred. However, as shown in Table~\ref{tab:WER_noVC}, AWS has a  higher \WER than Google. Hence, using both the \CSPs results in lower overall utility than just using Google's cloud service. 
}

\label{sec:topicModel}

\subsection{Voice Conversion}\label{sec:VoiceConversion}
\begin{figure}
    \centering
    \includegraphics[width=\columnwidth]{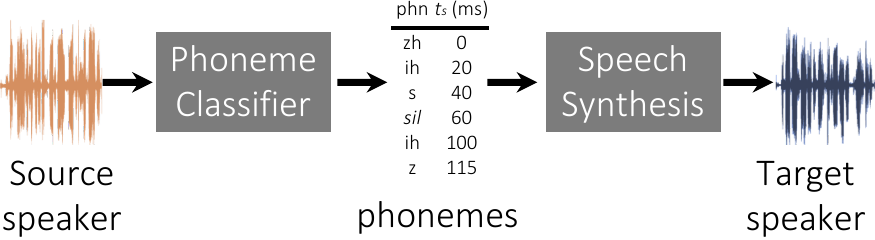}
    \caption{An illustration of the many-to-one \VC pipeline.}
    \label{fig:VC many-to-one}
\end{figure}

Below, we discuss the two main categories of \VC systems, highlighting their privacy-utility trade-offs.

\subsubsection{One-to-One Voice Conversion} 
\label{sec:one-to-one VC}
One-to-one \VC maps a predefined source speaker voice to a target speaker voice. 
In \name, we use sprocket~\cite{kobayashi2018sprocket}, which is based on spectral conversion using a Gaussian mixture model (GMM).
Sprocket's training phase takes three steps: (1) acoustic features extraction of the source and target speakers samples, (2) time-alignment of the source and target features, and (3) GMM model training. During conversion, sprocket extracts the acoustic features of the new utterances, converts them using the learned GMM model, and generates the target waveform. \name applies sprocket to convert the voice of all source speakers, including the synthesized dummy segments, into the \textit{same target speaker voice}. 

\subsubsection{Many-to-One Voice Conversion}
For perfect voice privacy, the \VC system should 
(1)  map any voice (even if previously unseen) to the same target voice, (2)  not leak any distinguishing acoustic features, and (3) operate on speech containing multiple speakers. To this end, \name deploys the two-stage many-to-one \VC ~\cite{VC} mechanism. As shown in Fig.~\ref{fig:VC many-to-one}, the first stage is a phoneme classifier that transfers the speech utterance into phonetic posterior grams (PPG) matrix. A PPG is a time-aligned phonetic class ~\cite{VC}, where a phoneme is the visual representation of a speech sound. Thus, the phoneme classifier removes the speaker-identifying acoustic features by mapping the spoken content into speaker-independent labels. In the second stage, a speech synthesizer converts the PPGs into the target voice.

The PPGs intermediate stage is irreversible and speaker-independent. It guarantees that the converted dummy segments $\mathbb{S}_d$ and converted original segments $\mathbb{S}$ cannot be distinguished from each other.
However, the actual implementation of the system carries many challenges. The first stage is a performance bottle-neck as it needs large phonetically aligned training data to generalize to new unseen voices. We overcome this challenge by generating a custom training speech dataset with aligned phonemes as described in Sec.~\ref{sec:implementaton}.

{\subsubsection{Control Knobs } \label{sec:VC:control_knobs}
The aforementioned \VC techniques present an interesting utility-usability-privacy trade-off. The one-to-one \VC technique gives better accuracy than many-to-one \VC since it is trained for a specific predefined set of source speakers (details in Sec. \ref{sec:Q4_utility_privacy}). However, this utility gain comes at the price of usability and privacy.   First, unlike many-to-one \VC, sprocket needs parallel training data -- a set of utterances spoken by both the source and target speakers. Hence, it requires an enrollment phase to get the source speaker's voice samples, thereby limiting the scalability of \name for previously unseen speakers. Second, one-to-one \VC does not provide perfect indistinguishability. These two limitations are mitigated by applying many-to-one \VC (Sec. \ref{sec:Q4_utility_privacy}).}

{\section{End-to-End Threat Analysis}
\label{sec:ThreatModel}
In this section, we go over the end-to-end system design of \name and identify potential privacy vulnerabilities.\\\\
\textbf{Voice Privacy:} Many-to-one \VC removes all the identifying features from $S$, like the speakers' voices, background noise, and recording hardware, thereby protecting voice privacy. } \\\\
{\textbf{Textual Privacy:} For sensitive word scrubbing, the best-case scenario from a privacy point of view is to have the user spell out the entire keyword list. However, due to its high usability overhead, \name uses \NER instead to identify named entities automatically from $T^{OSP}_S$. In Sec.~\ref{sec:eval:sws}, we empirically show that \name can achieve near-perfect true positive rate in identifying the segments containing sensitive words. However, this is only an empirical result and is dataset dependent.} \par 
Our main defense against statistical analysis on the text is the \DP guarantee on the word histogram.  This \DP guarantee would break down if the adversary can distinguish the dummy segments from the true segments. Many-to-one \VC technique, by design, ensures that both sets of segments have the same acoustic features. \minrev{However, the possibility of distinguishing them based on their textual features still remains. To address this threat, we rely on state-of-the-art \NLP models with low perplexity (log-likelihood) scores to generate the dummy text corpus. The low perplexity scores ensure that the auto-generated text is as close as possible to the natural language generated by humans~\cite{hofmann2013probabilistic, gpt2}. 
Although there is no formal guarantee about the adversary's ability to distinguish dummy and true segments based on their textual features, we have empirically analyzed this threat in Sec.~\ref{sec:eval:statisticalAnalysis} and Sec.~\ref{sec:indistinguishability_Dummy}.
}
We leverage state-of-the-art \NLP techniques to mount attacks on the dummy segments. Our results show that the adversary fails to distinguish between the dummy and true segments. However, the extent of such robustness is based on the efficacy of state-of-the-art \NLP techniques. 
\par Word correlations can also weaken the DP guarantee ($d-w$, if $w$ is the maximum size of word groups with high correlation). This can be addressed by either increasing $d$ or considering $n$-gram $(n=w)$ word histograms. However, this would increase the requisite amount of dummy segments.

\par  {Long segments can also be a source of privacy vulnerability as each segment contains more contextual information.
Hence, in the prototype \name presented in the paper, we use short segments that contain at most two non-stop words. 
}
\par Another weakness is related to vocabulary estimation, especially if some of the distribution-tail words are deemed to be sensitive. \minrev{\name provides no formal guarantees on the words that do not belong to $\mathcal{V}$.}  Although our empirical evaluation shows that the \OSP has a very high accuracy for the weighted estimation of $\mathcal{V}$ (Sec.~\ref{sec:DPAnalysis}), some sensitive distribution-tail words might still be missed due to the \OSP's transcription errors.  \minrev{Additionally, our formal DP guarantee holds only for the word histogram ($BOW$) on $\mathcal{V}$.
Textual analysis models other than $BOW$ are empirically evaluated in Sec.~\ref{sec:eval:statisticalAnalysis} and Sec.~\ref{sec:indistinguishability_Dummy}.}
\par
{Finally,  if the \CSP can reorder the segments (even partially since the speech file it receives contains dummy segments as well), it will be able to distinguish the dummy segments from the true ones and hence, learn the textual content of the file.  For this again, we show empirically that current \NLP techniques fail to reorder the segments (Sec.~\ref{sec:indistinguishability_Dummy}) even in the worst-case setting where all the segments go to one \CSP. However, as before, this is an empirical result only.}

\paragraph{Formal Privacy Guarantee:}
\textit{For a speech file $S$, \name provides perfect voice privacy (when using many-to-one \VC) and an $(\epsilon,\delta)$-\DP guarantee on the word histogram for the vocabulary considered ($BOW$), under the assumption that the dummy segments are indistinguishable from the true segments.}

\section{Implementation }
 
\label{sec:implementaton}

In this section, we describe the implementation details of \name's building blocks (shown in Fig.~\ref{fig:high-level}).

\paragraph{Segmentation:}
We implement the two-level hierarchical segmentation algorithm described in Sec.~\ref{sec:Segmentation}. The silence detection based segmentation is implemented using the Python pydub package\footnote{https://pypi.org/project/pydub/}. 
We used Praat\footnote{http://www.fon.hum.uva.nl/praat/} to extract the pitch information required for the second level of the segmentation algorithm.

\paragraph{Sensitive Keyword Scrubbing:}
 
{We use the NLP Python framework spaCy~\footnote{https://github.com/explosion/spaCy} for named entity recognition (NER) from the text. 
The keyword lists per each dataset can be found in \ifpaper the full paper \cite{FP} \else Appendix B\fi.}
We employ PocketSphinx\footnote{https://github.com/cmusphinx/pocketsphinx} for keyword spotting, \minrev{a lightweight \ASR that can detect
keywords from continuous speech.  It takes a list of words (in the text) and their respective sensitivity thresholds and returns segments that contain speech matching the words. PocketSphinx is a generic system that can detect any keyword specified in runtime; it is not trained on a pre-defined list of keywords and requires no per-user training or enrollment.}

\paragraph{Generating Dummy Segments:}
{
We use the open source implementation~\footnote{https://github.com/huggingface/transformers} of OpenAI's state-of-the-art NLP language model, $GPT2$\cite{gpt2}, to generate the noise corpus. 

Using this predictive model, we generate a large corpus representing the vocabulary of the evaluation datasets. An example of the generated text is available in \ifpaper the full paper \cite{FP} \else  Appendix B\fi.} 
To generate the dummy segments, we segment each document at the same level as the speech segmentation algorithm. We build a hash table associating each vocabulary word with the segments that contain it. \name uses a dummy segment only once per \CSP to prevent it from identifying repetitions.

\paragraph{Text-to-Speech:}
We use the multi-speaker (voice cloning) \TTS synthesizer~\cite{google_speech_synthesis} to generate the speech files corresponding to the dummy segments. We use a pre-existing system implementation and pretrained models~\footnote{https://github.com/CorentinJ/Real-Time-Voice-Cloning}. 

\noindent
{\paragraph{One-to-One Voice Conversion:}
We use the open-source sprocket software~\footnote{https://github.com/k2kobayashi/sprocket}. As described in Sec.~\ref{sec:one-to-one VC}, sprocket requires a parallel training data and the target voice should be unified for all source speakers. For the VCTK datasets, we use speaker p306 as the target voice. Since we also evaluate \name on non-standard datasets (Facebook and Carpenter cases), we had to construct the parallel training data for their source speakers. For this, we use \TTS to generate the required target voice training utterances in a single \textit{synthetic} voice.}

\paragraph{Many-to-One Voice Conversion:}
We utilize pre-existing architectures and hyperparameters~\footnote{https://github.com/andabi/deep-voice-conversion} for the two-stage many-to-one \VC ~\cite{VC} mechanism, shown in Fig.~\ref{fig:VC many-to-one}. The first network, $net_1$, is trained on a set of $\{$raw speech, aligned phoneme labels$\}$ samples from a multi-speaker corpus, where the labels are the set of 61 phonemes from the TIMIT dataset. The only corpus that has a manual transcription of speech to the phonemes' level is the TIMIT dataset -- a limited dataset. We found that training $net_1$ on TIMIT alone results in an inferior \WER performance. For better generalization, we augment the training set by automatically generating phoneme-aligned transcriptions of standard \ASR corpora. We use the Montreal Forced Aligner~\footnote{https://montreal-forced-aligner.readthedocs.io/en/latest/} to generate the aligned phonemes on LibriSpeech and TED-LIUM~\cite{rousseau2012ted} datasets.
The second network, $net_2$, synthesizes the phonemes into the target speaker's voice. It is trained on a set of $\{$PPGs, raw speech$\}$ pairs from the target speaker's voice. We use the \textit{trained} $net_1$ to generate the PPGs data for training $net_2$. As such, we only need speech samples of the target speaker to train $net_2$. This procedure also allows $net_2$ to account for $net_1$'s errors. We use Ljspeech\footnote{https://keithito.com/LJ-Speech-Dataset/} as the target voice for its relatively large size -- 24 hours of speech from a single female.

\section{Evaluation}
\label{sec:evaluation}

\begin{table*}[t]
\begin{minipage}[b]{0.60\linewidth}\centering
\scalebox{0.95}{
\begin{tabular}{@{}lccccc@{}}\toprule
\textbf{Datasets} & {\textbf{Cloning}} & {\textbf{One-to-One}} & \textbf{Many-to-One} & \textbf{\OSP}\\ 
\midrule
VCTK p266 & {5.15 (80.73\%)}  & {16.55 (38.06\%)}  & 21.92 (17.96\%) & 26.72\\
VCTK p262 & {4.53 (71.63\%)}& {7.39 (53.73\%)} & 10.82 (32.25\%) & 15.97\\
\midrule
Facebook1 & {8.26 (66.59\%)} & {14.60 (40.94\%)} & 20.30 (17.88\%) & 24.72\\
Facebook2 & {9.75 (63.36\%)} & {18.27 (31.34\%)} & 19.44 (26.94\%) & 26.61\\
Facebook3 & {14.93 (51.40\%)}& {23.25 (24.32\%)} & 27.06 (11.91\%) & 30.72\\
\midrule
Carpenter1 & {14.43 (44.18\%)} & {23.88 (7.62\%)} & 22.63 (12.46\%) & 25.85\\
Carpenter2 & {13.53 (65.93\%)} & {33.71 (15.11\%)} & 38.90 (2.04\%) & 39.71\\
\bottomrule
\end{tabular}}
\caption{WER (\%) of end-to-end \name which represents the accumulative effect of segmentation, \SWS, and different settings of voice privacy and its relative improvement in (\%) over \OSP (Deep Speech).}  
\label{tab:WER_VC_4levels}
\end{minipage}
\hfill
\begin{minipage}[b]{0.37\linewidth}
\centering
\includegraphics[width=0.95\textwidth]{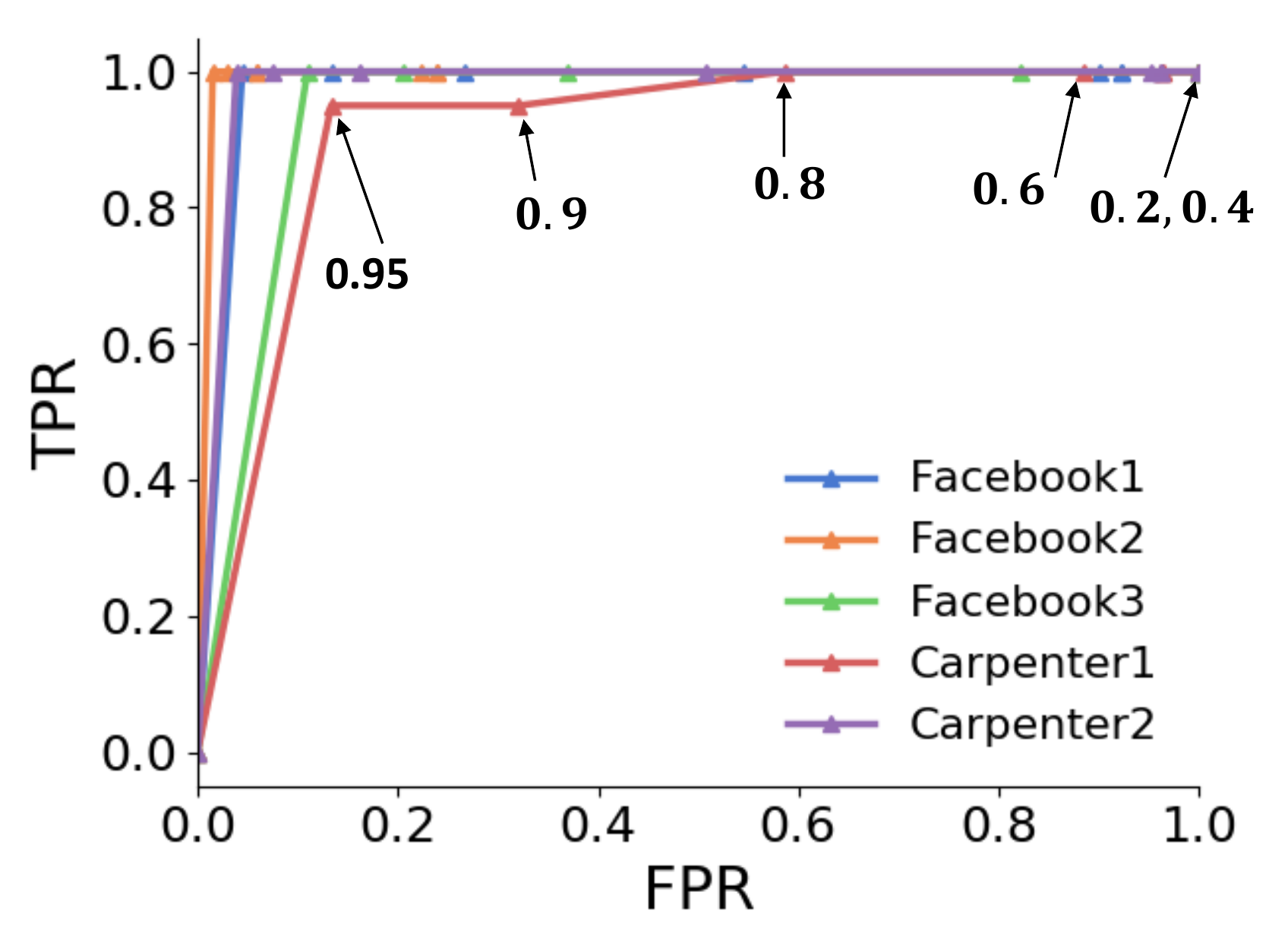} \vspace{-2pt}
\captionof{figure}{{ROC curve for sensitive words detection at different values of the sensitivity score.}}
\label{fig:KWS}
\end{minipage}
\end{table*}

We evaluate how well \name meets the design objectives of Sec.~\ref{sec:OurSys}. Specifically, we aim to answer the following questions: (\textit{Q1.}) Does \name preserve the transcription utility?\\ (\textit{Q2.}) Does \name protect the speakers' voice biometrics? \\(\textit{Q3.}) Does \name protect the textual content of the speech? \\ {(\textit{Q4.}) Does the different control knobs provide  substantial flexibility in the utility-usability-privacy spectrum?}  

{We answer the first three questions for a prototype implementation of \name that provides the maximum degree of formal privacy and hence, the least utility. For evaluating \textit{Q4}, we relax the privacy guarantee to obtain utility and usability improvements.}

{\paragraph{Prototype \name:} For the prototype \name presented in the paper:  (1) segmentation length is adjusted to ensure that each segment contains at most two non-stop words (2) noisy segments are generated via the GPT2  language model 
(3) a single \CSP (Google) is utilized (4) many-to-one \VC is applied to both the dummy and true segments.  
}
\subsection{Q1. Transcription Utility} \label{sec:transcriptionUtility}

We assess the transcription \WER after deploying end-to-end \name on the non-standard datasets.
Recall that Table~\ref{tab:WER_noVC} in Sec.~\ref{sec:accuracy_comparison} shows the baseline \WER performance of the \CSP and \OSP before applying \name. 

\paragraph{\WER Analysis:}
Column 4 in Table~\ref{tab:WER_VC_4levels} shows the end-to-end \WER for the prototype \name which represents the accumulative effect of segmentation, \SWS, and many-to-one \VC. Although \VC is the main contributor to \name's \WER, as is evident from Sec.~\ref{sec:Q4_utility_privacy} and Sec.~\ref{sec:eval:sws}, there are two main observations. First, many-to-one \VC is superior to  Deep Speech. Specifically, \name's relative improvement over Deep Speech ranges from 11.91\% to 32.25\% over the evaluation datasets (except for Carpenter2). Recall that we trained the \VC system using standard \ASR corpora, while we evaluate the \WER on non-standard cases. Still, \name's \WER is superior to that of Deep Speech, which has been trained through hundreds of hours of speech data. Second,  \name does not have the same performance for all the datasets. This observation arises again from the lack of diversity in our \VC training set. For example, the speaker in Carpenter 1 speaks loudly, allowing \VC to perform well. On the other hand, the second speaker (Carpenter 2) is not as clear or loud, which results in an inferior \VC performance. This observation is consistent with Deep Speech as well. \\
Our experiments show that these results can be improved by adding samples of the source speaker voice to the training pipeline of $net_1$ and $net_2$. We chose not to go with this approach as this limits the usability of the system, and in such a case sprocket (Sec.~\ref{sec:Q4_utility_privacy}) would be a better choice.

\subsection{Q2. Voice Biometric Privacy}
\label{sec:voice_eval}

To test the voice biometric privacy, we conduct two experiments using the voice analysis APIs (details in Sec.~\ref{sec:voice_threat}).
In the first experiment, we assess the \CSP's ability to separate speech belonging to different speakers after \name applies the \VC system. On our multi-speaker datasets, IBM diarization API concludes that there is only one speaker present. 

Furthermore, we run the diarization API after adding the dummy segments (after \TTS and \VC). Again, the API detects the presence of only one speaker. Thus, not only does \name hide the speaker's biometrics and map them to a single target speaker but also ensures noise indistinguishability, which is key to its privacy properties. 

The second experiment tests \name's privacy properties against a stronger adversary, who has access to samples from the true speakers. We enroll segments from the true speakers as well as the fake target speaker to Azure's Speaker Identification API. We pass the segments from \name (after adding dummy segments and applying \VC) to the API. 
When many-to-one \VC is applied, in all evaluation cases, the API identifies the segments as belonging to the fake target speaker. Not a single segment was matched to the original speaker. Both experiments show that prototype \name is effective in sanitizing the speaker's voice and ensuring noise indistinguishability. 

\begin{table*}[t]
\small
\begin{minipage}[b]{0.450\linewidth}\centering
\scalebox{1}{\begin{tabular}{@{}lccccc@{}}\toprule
\multirow{2}{*}{\textbf{Datasets}} & \multirow{2}{*}{\textbf{$|\mathcal{V}|$}} & {\multirow{1}{*}{\textbf{\# words}}} & \multicolumn{3}{c}{\textbf{\#Extra words due to dummy segments}}\\
 & & in $T^g_S$ & \textbf{d=2} & \textbf{d=5} & \textbf{d=15}\\ 
\midrule
VCTK p266 & 483 & 922 (\$0.22) & 2915 (\$0.68) & 7247 (\$1.69) & 23899 (\$5.58) \\
VCTK p262& 471 & 914 (\$0.21)  & 2845 (\$0.66) & 7157 (\$1.67) & 23230 (\$5.42)\\
Facebook & 1098 & 5326 (\$1.24) & 6660 (\$1.55) & 16567 (\$3.87) & 54038 (\$12.62)\\
Carpenter & 1474 & 7703 (\$1.80) & 8915 (\$2.08) & 22296 (\$5.20) & 72907 (\$17.02)\\
\bottomrule
\end{tabular}}
\caption{Number of extra words due to dummy segments and the additional monetary cost in USD with varying $d$, at $\epsilon=1$ and $\delta=0.05$.}
\label{tab:noise_cost}
\end{minipage}
\hfill
\begin{minipage}[b]{0.36\linewidth}
\centering
\includegraphics[width=0.82\textwidth, height=3.7cm]{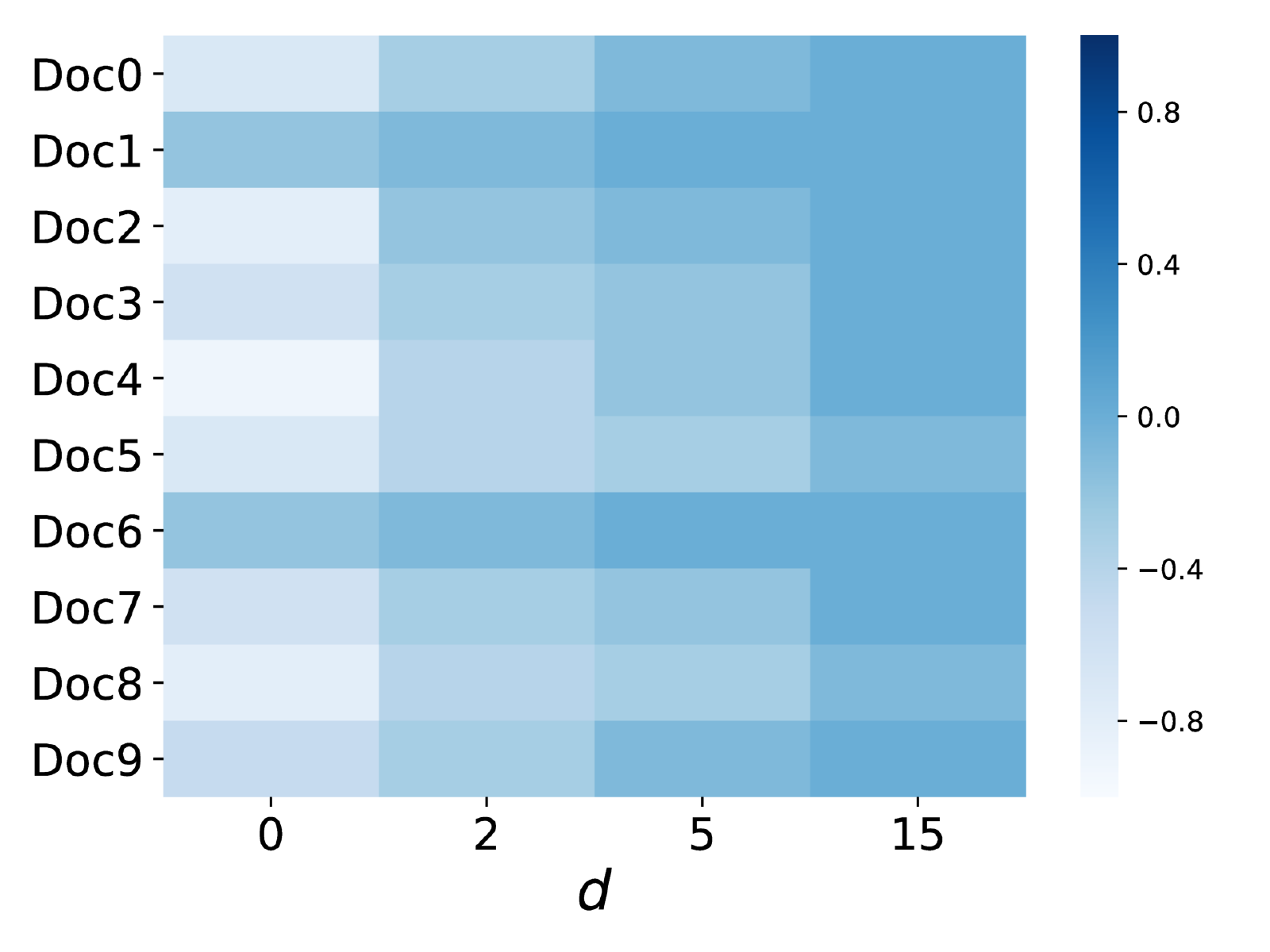}
\vspace{-4pt}
\captionof{figure}{Sentiment scores heatmap of 10 documents with varying $d$, at $\epsilon=1$ and $\delta=0.05$.}
\label{fig:heatmap}
\vspace{-4pt}
\end{minipage}
\end{table*}
\vspace{-0.23cm}

\subsection{Q3. Textual Privacy}
\label{sec:eval:text}

We perform an extensive evaluation of the textual privacy, including sensitive word scrubbing, analysis of the \DP mechanism, and defense against statistical analysis.

{\subsubsection{Sensitive Words Scrubbing:}\label{sec:eval:sws}
We run PocketSphinx keyword spotting on each dataset at different sensitivity scores ranging from 0.2 to 1\footnote{The sensitive keywords list for each dataset is in \ifpaper the full paper \cite{FP} \else Appendix B\fi.}. Fig.~\ref{fig:KWS} shows the detection true positive rate (\TPR) versus the false positive rate (\FPR) 
at different sensitivity scores. As the figure shows, the sensitivity score is a trade-off knob between privacy (high \TPR) and utility (low \FPR). We observe that \name is able to achieve almost perfect \TPR with low \FPR values.}  

{Next, we evaluate the impact of \SWS on the transcription utility. We set a sensitivity score of 0.95 for all the datasets to have a near-perfect \TPR while minimizing the \FPR. Our experiments show that the total duration of the segments flagged with sensitive keywords at this score is: 0.13\%, 0.06\%, 0.18\%, 0.20\%, and 0.08\% of the total duration of each dataset in Fig.~\ref{fig:KWS}. Then, we transcribe the sensitive-flagged segments using Deep Speech. The overall transcription accuracy after \SWS (i.e., equivalent to choosing voice cloning in \name as cloning results in no addition \WER)  is presented in the second column of Table~\ref{tab:WER_VC_4levels}. 
Since the segments are short, the portion of speech transcribed locally is limited. Hence, the impact of  the \OSP transcription errors is not significant.}

\subsubsection{DP Mechanism Analysis:}\label{sec:DPAnalysis}
We follow the DP mechanism described in Sec.~\ref{sec:DPmechanism}. 

\paragraph{Vocabulary Estimation:}
We estimate the vocabulary $\mathcal{V}$ using the \OSP transcript. Let $\mathcal{W}$ represent the set of unique words in $T^g_S$. We define the accuracy of the vocabulary estimation, $D_{acc}$, as the ratio between the count of the correctly identified unique words from $T^{OSP}_S$, $|\mathcal{W}|_{est}$, and the count of the unique words in $T^g_S$, $|\mathcal{W}|$. For our datasets, the domain estimation accuracy is at least $75.54\%$. We also calculate the weighted estimation accuracy defined as: $
D_{weighted} = \frac{\sum{P(w_{est}).\mathds{1}_{w_{est} \in \mathcal{W}}}}{|\mathcal{W}|}
$ where $P(w_{est})$ is the weight of the estimated word $w_{est}$ in $T^g_S$. $D_{weighted}$ is more informative since it gives higher weights to the most frequent words in $T^g_S$. The weighted estimation accuracy is $99.989\%$ in our datasets. From $\mathcal{W}_{est}$ we select $\mathcal{V}$ over which we apply the DP mechanism.  Additionally, we extend our vocabulary to contain a set of random words from the English dictionary.

\paragraph{Histogram Distance:}
 
We analyze the distance between the original and noisy histograms (after applying \name) and its impact on the cost of online transcription. Because of the nature of \name's DP mechanism, the noise addition depends on four values only: $|\mathcal{V}|$, $\epsilon$, $\delta$, and $d$.

For all our experiments, we fix the values of $\epsilon=1$ and $\delta=0.05$. {Table~\ref{tab:noise_cost} shows the amount of noise (dummy words) and their transcription cost in USD~\footnote{The pricing model of Google Speech-to-Text is: \$0.009 / 15 seconds.} for each of the evaluation datasets at different values of $d$. Each dataset has a different vocabulary size $|\mathcal{V}|$ and word count.}  
The increase in the vocabulary size requires adding more dummy segments to maintain the same privacy level. In \name, adding more noise comes at an increased monetary cost, instead of a utility loss. The table highlights the \textit{trade-off} between privacy and the cost of adding noise.

\begin{figure*}     
\centering
\begin{subfigure}{0.24\linewidth}
  \centering
  \includegraphics[width=\linewidth]{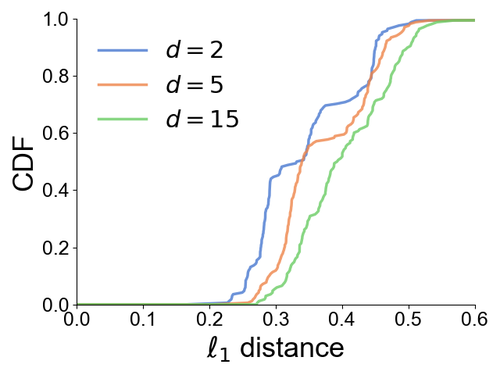}
  \caption{8 topics}
  \label{fig:cloud_1} 
\end{subfigure}
\begin{subfigure}{0.24\linewidth}
  \centering
  \includegraphics[width=\linewidth]{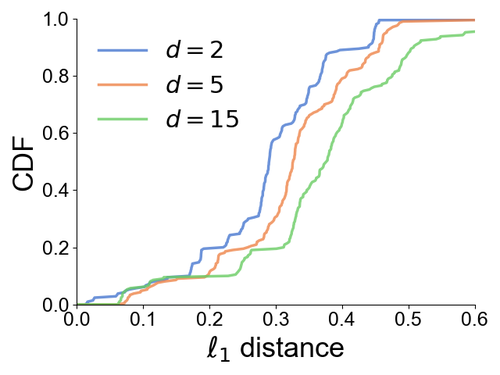}
  \caption{10 topics}
  \label{fig:cloud_2} 
\end{subfigure}
\begin{subfigure}{0.24\linewidth}
  \centering
  \includegraphics[width=\linewidth]{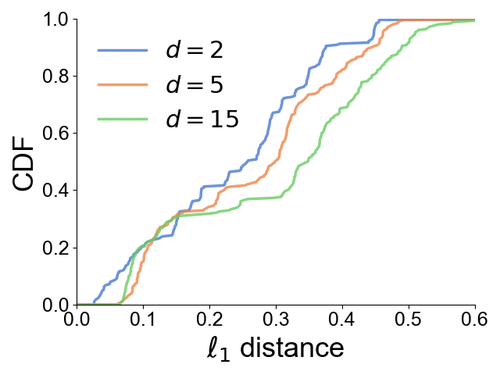}
  \caption{12 topics}
  \label{fig:cloud_3} 
\end{subfigure}
\begin{subfigure}{0.24\linewidth}
  \centering
  \includegraphics[width=\linewidth]{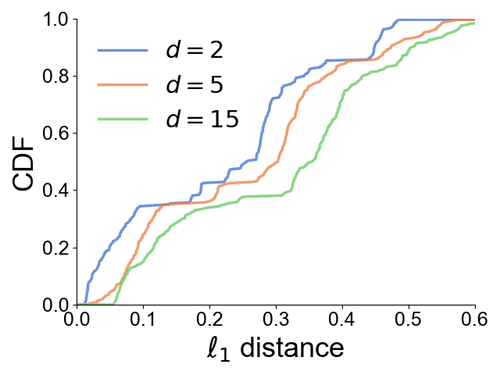}
  \caption{14 topics}
  \label{fig:cloud_4} 
\end{subfigure}
\caption{Topics $\ell_1$ distance \CDF at $d$ = 2, 5, and 15 for $t=$ 8, 10, 12, and 14}
\label{fig:topic_distance}
\end{figure*}
 
\subsubsection{Statistical Analysis}\label{sec:eval:statisticalAnalysis}
In this section, we evaluate the statistical analyses (details in Sec.~\ref{sec:textAnalysis}) performed by the adversary to extract textual information on the noisy transcripts obtained from \name. 
 
\paragraph{Topic Model:}
We generate the topic models from the documents corresponding to the original and noisy word histograms, and evaluate their $\ell_1$ distance. The topic model operates on a corpus of documents; hence we include eight more Supreme Court cases to our original evaluation datasets (Facebook and Carpenter). 
In this evaluation, we treat all these ten documents as one corpus; we aim to generate the topic model before and after applying \name to the whole corpus.

We use AWS Comprehend API to generate the topic model. The API needs the number of topics as a hyperparameter that ranges from 1 to 100. Based on our apriori knowledge of the true number of topics, we evaluate the topic model on the following number of topics $t = 8, 10, 12, $ and $14$.

We statistically evaluate the $\ell_1$ distance between true and noisy topics. The topic model $\mathcal{T}=\{\mathbb{T}_1,\cdots,\mathbb{T}_t\}$ is a set of $t$ topics where each $\mathbb{T}_i, i \in [t]$ is a word distribution. We use the Hungarian algorithm to match each noisy topic $\mathbb{ T}' _i \in \mathcal{ T}' $ to its closest match in $\mathcal{T}$, the true topic model. We evaluate the topics $\ell_1$ distance for 21 runs. At each run, we generate a random noise vector per document, select the corresponding dummy segments, and evaluate the topic model on the set of original and noisy documents.
Fig.~\ref{fig:topic_distance} shows the empirical \CDF of the topics $\ell_1$ distance at different values of $d$. As the figure shows, the higher the distance parameter $d$, the larger is the $\ell_1$ distance between true and noisy topics.
 
\noindent
{\paragraph{Stylometry:}
 
In this experiment, we assume that the \CSP applies stylometry analysis on $T^{\CSP}_S$ in an attempt to attribute it to an auxiliary document whose authors are known to the \CSP. To evaluate the worst-case scenario, we assume the adversary possesses the original document $T^{g}_S$, and we compute the $\ell_2$ distance of the stylometric feature vectors generated from $T^{\CSP}_S$ w.r.t $T^{g}_S$.

First, we compute the $\ell_2$ distance of $T^{\CSP}_S$ before applying \name. The respective values for the Facebook and Carpenter datasets are 28.19 and 60.45.  $T^{\CSP}_S$ differs from $T^{g}_S$ in  lexical features due to transcription errors and because the \CSP generates the punctuation  instead of the actual author.\\
Second, we apply \name on the two datasets at different values of the distance parameter: $d= 0, 2, 5, 15$. The corresponding $\ell_2$ distances for the Facebook (Carpenter) dataset equal: 73.14 (83.64), 328.80 (577.72), 947.58 (1629.79), and 2071.18 (3582.10). Note that the $\ell_2$ distance at $d=0$ shows the effect of segmentation and \SWS only on obfuscating the lexical features. Clearly, adding the dummy segments increases the $\ell_2$ distance.  This is expected as most of the lexical features are obfuscated by the \DP mechanism. 
} 
\noindent
\paragraph{Category Classification:} 
 Google's NLP API can classify a document to a predefined list of 700+ document categories\footnote{https://cloud.google.com/natural-language/docs/categories}. First, we run the classification API on the original documents from the topic modeling corpus. All of them classify as \textsf{Law \& Government}\xspace. Running the API on \name processed documents, using an extended-vocabulary (i.e., contains random words), dropped the classification accuracy to $0\%$. None of the documents got identified as legal, law, or government even at the smallest distance parameter value $d=2$.  
Although a portion of the noise words belongs to the original  \textsf{Law \& Government}\xspace category, segmentation, shuffling, and the out-of-domain noise words successfully confuse the classifier.

\paragraph{Sentiment Analysis:}
Sentiment analysis generates a score in the $[-1, 1]$ range, which reflects the positive, negative, or neutral attitude in the text. 
First, we evaluate the sentiment scores of the original ten documents.
For all of them, the score falls between $-0.2$ and $-0.9$, which is expected as they represent legal documents. Next, we evaluate the scores from \name processed documents considering an extended-vocabulary. We find that all scores increase towards a more positive opinion.  Fig.~\ref{fig:heatmap} shows a heatmap of the sentiment scores as we change the distance parameter $d$ for the then evaluation documents.
Thus, \name's two-pronged approach -- 1) addition of extended-vocabulary noise 2) removal of ordering information via segmentation and shuffling, proves to be effective. In a setting where the adversary has no apriori knowledge about the general domain of the processed speech, the noise addition mechanism gains extend from DP guarantee over the histogram to other NLP analyses as well.
 
{\subsubsection{Indistinguishability Of Dummy Segments}\label{sec:indistinguishability_Dummy}
The indistinguishability of the dummy segments is critical for upholding the DP guarantee in \name.
We perform two experiments to analyze whether current state-of-the-art NLP models can distinguish the dummy segments from their textual content.}
\noindent
{\paragraph{Most Probable Next Segment:}
In this experiment, the adversary has the advantage of knowing a true segment $\mathbb{S}_t$ that is at least a few sentences long from the Facebook dataset.  
We use the state-of-the-art GPT~\footnote{https://github.com/huggingface/transformers} language model by OpenAI~\cite{gpt} to determine the most probable next segment following $\mathbb{S}_t$ using the model's perplexity score. In NLP, the perplexity score measures the likelihood that a piece of text follows the language model. We get the perplexity score of stitching $\mathbb{S}_t$ to each of the other segments at the \CSP. The segment with the lowest perplexity score is selected as the most probable next segment. We iterate over all the true segments of the Facebook dataset, selecting them as $\mathbb{S}_t$. We observed that a dummy segment is selected as the most probable next segment in 53.84\% of the cases.
This result shows that the language model could not differentiate between the true and dummy segments even when part of the true text is known to the adversary.}
\noindent
{\paragraph{Segments Re-ordering:} Next, we attempt to re-order the segments based on the perplexity score. We give the adversary the advantage of knowing the first true segment $\mathbb{S}_0$. We get the perplexity score of $\mathbb{S}_0$, followed by each of the other segments. The segment with the lowest score is selected as the second segment $\mathbb{S}_1$ and so on. We use the normalized Kendall tau rank distance $K_\tau$ to measure the sorted-ness of the re-ordered segments. The normalized $K_\tau$ distance measures the number of pairwise disagreements between two ranking lists, where 0 means perfect sorting, and 1 means the lists are reversed. The $K_\tau$ score for running this experiment on the Facebook dataset is 0.512, which means that the re-ordered list is randomly shuffled w.r.t the true order. Hence, our attempt to re-order the segments has failed. }\\
\noindent
{These empirical results show that it is hard to re-order the segments or distinguish the dummy segments. This is expected due to three reasons: (1) the segments are very short; (2) the dummy segments are generated using a state-of-the-art language model; and (3) we observed that most of the transcription errors happen in the first and last words of a segment due to breaking the context. These errors add to the difficulty of re-ordering. Moreover, if the user partitions $S$ among multiple \CSP's (Sec.\ref{sec:DPmechanism}), then consecutive segments would not go to the same \CSP with high probability. This setting would increase \name's protection against re-ordering attacks.}
{\subsection{Q4: Flexibility of the Control Knobs }\label{sec:Q4}}

 {\subsubsection{Utility-Privacy Trade-off}\label{sec:Q4_utility_privacy}
 In this section, we empirically evaluate the controls knobs that provide a utility-privacy trade-off.}
\begin{figure}[t]     
\centering
\begin{subfigure}{0.49\linewidth}
  \centering
  \includegraphics[width=\linewidth]{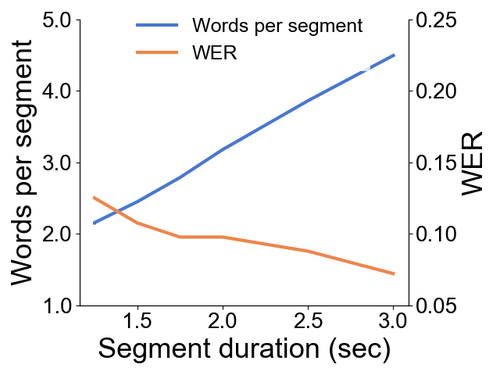}
  \caption{Facebook}
  \label{fig:mark_grid} 
\end{subfigure}
\begin{subfigure}{0.49\linewidth}
  \centering
  \includegraphics[width=\linewidth]{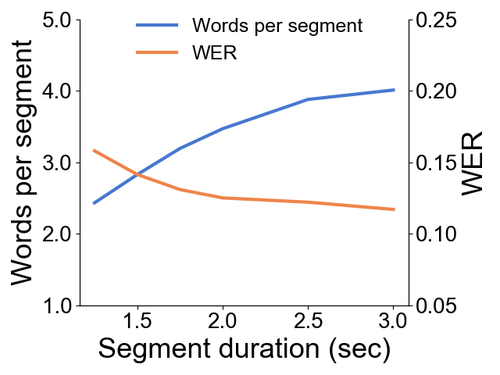}
  \caption{Carpenter}
  \label{fig:carpenter_grid} 
\end{subfigure}
\caption{Segmentation trade-off between utility and privacy. WER(\%) is measured using Google Cloud Speech-to-Text.}
\label{fig:seg_trade_off}
\end{figure}
{\paragraph {Minimum
segment length:} 
Fig.~\ref{fig:seg_trade_off} shows the trade-off between the number of words per segment and \WER as function of the minimum segment length. As expected, increasing the minimum duration of a segment results in an increase in the number of words per segment. The \WER in turn drops when the number of words per segment increase as the transcription service has more textual context. However, it can lead to potential privacy leakage. The results in Fig.~\ref{fig:seg_trade_off} indicate that for two real-world datasets, the number of words per segment can be kept between 2 and 3 with an acceptable degradation of the \WER. 
} 
\noindent
{\paragraph{Voice Cloning:} Voice cloning does not affect the true segments (it is only applied to dummy segments), resulting in no additional \WER degradation. The \WER for deploying voice cloning is incurred only due to segmentation and \SWS. Thus, as shown in column 2 of Table ~\ref{tab:WER_VC_4levels}, the relative improvement in \WER ranges from 44\% to 80\% over Deep Speech.} This approach, however, has two limitations. First, the speaker's voice biometrics from $S$ are not protected. 
\minrev{Second, there is no guarantee that an adversary would not be able to distinguish the cloned speech segments from the original ones.}   
\noindent
{\paragraph{Sensitivity score of \KWS:} As shown in  Fig.~\ref{fig:KWS}, lower the sensitivity score, higher is the \TPR and hence greater is the privacy (most prominent in the Carpenter2 dataset). 
However, this also increases the FPR, which means a larger number of non-sensitive segments are transcribed via the \OSP resulting in reduced accuracy. }
\noindent
\paragraph{One-To-One \VC:} 
Table~\ref{tab:WER_VC_4levels}, column 3, shows that one-to-one \VC outperforms many-to-one \VC on most of the datasets. 
This result is expected since sprocket is trained and tested on the same set of source speakers while the many-to-one \VC system generalizes to previously unseen speakers. \\
We observe that the improvement for the VCTK dataset is more significant than others. Recall that in our one-to-one \VC implementation in Sec.~\ref{sec:implementaton}, the target voice for VCTK is a natural voice -- speaker p306. The target voice for the other datasets is a synthetic one, which hinders the quality of the converted voice and the transcription accuracy.
We investigate this observation by training sprocket for VCTK on a synthetic target voice as well. The \WER then increased to 19.33\% and 9.21\% for p266 and p262. 
Hence, we attribute the difference in the relative improvement to the target voice naturalness. In practice, the target voice could easily be a natural pre-recorded voice, and the users are asked to repeat the same utterances at the enrollment phase.\\
However, the one-to-one \VC technique suffers from some privacy loss. 
The one-to-one \VC system translates the acoustic features from a source to a target speaker's voice. Hence, it may leak some features from the source speaker. We observed that one-to-one \VC is vulnerable to speaker identification analysis. Specifically, using Azure's Speaker Identification API,
 10\% of the voice-converted segments using sprocket were identified to their true speakers. 

{\subsubsection{Usability-Privacy Trade-off}
\label{sec:usability_privacy}}
{In our setting, usability can be measured along three dimensions: latency, monetary cost, and implementation overhead. However, we would like to stress that \name is not designed for real-time speech transcription. 
Hence, latency is not a primary concern for \name. Nevertheless, we include it in the following discussion for the sake of completeness.
}

{\paragraph{{Latency Evaluations:}} 
Note that all the operations of \name are performed on speech segments. Hence, the latency is linear in the number of segments. We evaluate the end-to-end system latency per segment (with length $\sim 6$s) for the \OSP, the \CSP, and \name; the latency values are 2.17s, 1.70s, and 14.90s, respectively. We observe that the overhead of \name is mostly attributed to the many-to-one \VC (11s per segment on average).  
When  voice cloning (or one-to-one \VC) is applied instead, \name's end-to-end per segment latency reduces to  3.90s (or 11.47s) at the expense of a privacy loss as discussed in Sec.\ref{sec:Q4_utility_privacy}.}

\paragraph{Vocabulary Size:} Considering a larger $\mathcal{V}$ (Sec.~\ref{sec:DPmechanism}) increases the scope of the \DP guarantee. For example, adding external words provides protection against statistical analysis like text classification (Sec.\ref{sec:eval:text}). However, larger $\mathcal{V}$ results in increased amount of dummy segments and hence, increased monetary cost (Table~\ref{tab:noise_cost}). For example, extending $\mathcal{V}$ by $\sim1000$ out-of-domain words for the Carpenter dataset incurred a total cost of \$25 at $d=15$. 
\noindent
\paragraph{Distance Parameter $d$:} As explained in Sec.~\ref{sec:DP:discussion}, larger the value of $d$, greater is the scope of privacy.  
However, the amount of required noise increases by $d$. For example, for the dataset VCTK p266, increasing $d$ from 2 to 15  increases the cost by roughly \$5 (Table~\ref{tab:noise_cost}).

{\subsubsection{Utility-Usability Trade-off}\label{sec:utility_usability}
The following control knobs provide a venue for customizing the utility-usability trade-off.  
\paragraph{Number of \CSPs:} As discussed in Sec.~\ref{sec:DP:discussion}, using multiple \CSPs reduces the amount of dummy segments (and hence, the monetary cost) in \name.  
However, it comes at the price of utility; the transcription accuracy of the different available \CSPs varies. For example, from Table \ref{tab:WER_noVC}, we observe that AWS has a higher \WER than Google. Thus, using multiple \CSPs may result in a lower mean utility.
\noindent \paragraph{One-to-One \VC:} As discussed above, one-to-one \VC technique has lower \WER than many-to-one \VC technique (Table \ref{tab:WER_VC_4levels}). However, it requires access to representative samples of the source speaker voice for parallel training thereby limiting scalability for previously unseen speakers (Sec.~\ref{sec:VoiceConversion}).}

\section{Related Work}
\label{sec:RW}
In this section, we provide a summary of the related work. 

\paragraph{Privacy by Design:}
One class of approaches redesigns the speech recognition pipeline to be private by design. For example, Srivastava et al. proposes an encoder-decoder architecture for speech recognition~\cite{srivastava:hal-02166434}. 
Other approaches address the problem in an SMC setting by representing the basic operations of a traditional \ASR system using cryptographic primitives~\cite{pathak2013privacy}. 
VoiceGuard is a system that performs \ASR in the trusted execution environment of a processor 
~\cite{brasser2018voiceguard}. However, these approaches require redesigning the existing systems. 

\paragraph{Speech Sanitization:}
Recent approaches have considered the problem from a similar perspective as ours. They sanitize the speech before sending it to the \CSP. 
One such approach randomly perturbs the MFCC, pitch, tempo, and timing features of a speech before applying speech recognition~\cite{vaidya2019you}. Others sanitize the speaker's voice using vocal tract length normalization (VTLN)~\cite{qian2017voicemask,qian2018towards}. 
A recent approach modifies the features relevant to emotions from an audio signal, makes them less sensitive through a GAN~\cite{aloufi2019emotionless}. 
Last, adversarial attacks against speaker identification systems can provide some privacy properties. These approaches apply minimal perturbations to the speech file to mislead a speaker identification network~\cite{cai2018attacking,Kreuk_2018}.

These approaches are different from ours in two ways. First, they do not consider the textual content of the speech signal. The only exception is the approach by Qian et al.~\cite{qian2018towards}, which addresses the problem of private publication of speech datasets. This approach requires a text transcript with the audio file, which is not the case for the speech transcription task.  In addressing the textual privacy of a speech signal, \name adds indistinguishable noise to the speech file. The proposed techniques
fail to provide this property. Second, the approaches above only consider voice privacy against a limited set of features, such as speaker identification or emotion recognition.  \name applies many-to-one \VC to provide perfect voice privacy.

\section{Conclusion}
\label{sec:conclusion}
In this paper, we have proposed \name, an end-to-end system for speech transcription that  (1) protects the users' privacy along the acoustic and textual dimensions at (2) an improved performance relative to offline \ASR, (3) while providing customizable utility, usability, and privacy trade-offs.

\section*{Acknowledgment}

The work reported in this paper was supported in part by the NSF under grants 1661036, 1838733, 1942014, and 1931364. We also acknowledge Google for providing us with Google Cloud Platform credits and NVIDIA Corporation with the donation of the Quadro P6000 GPU used for this research. We would like to thank the anonymous reviewers for their useful comments and Micah Sherr for shepherding this paper.

\bibliographystyle{abbrv}
\bibliography{_10_refs.bib}

\ifpaper \newpage 
\else
\newpage 
\section{Appendix A}
This appendix contains the proof of Theorem~\ref{thm:topic_bound} of Sec.~\ref{sec:DP}.

\begin{table}[ht]
\centering
\caption {Notations}
\scalebox{0.7}{
 \begin{tabular}{|l|l|} \toprule
 \multicolumn{1}{|c|}{\textbf{Symbol} } &  \multicolumn{1}{|c|}{\textbf{Explanations}}\\\hline
  $\mathcal{V}$ & - the vocabulary  \\
  $t$ & - total number of topics \\ 
  $w$ & - represents a word \\
  $k$ & - minimum number of words per topic \\
  $\mathbb{T}$ & - represents a true topic \\
  $\mathbb{T}'$ & - represents a perturbed topic \\
  $\mathcal{T}=\langle \mathbb{T}_1, \cdots, \mathbb{T}_t\rangle$ & - represents the true topic model \\
  $\mathcal{T}'=\langle \mathbb{T}'_1, \cdots, \mathbb{T}'_t\rangle$ & - represents the perturbed topic model\\
  $d$ & - chosen distance parameter \\
  $n$ &- total number of documents\\
  $D$ & - represents a document\\
  $\mathcal{D}=\cup_{i=1}^d D$ & - the corpus of documents\\ 
  $\omega_j$ & - total number of unique words in document $D_j$\\
  $|w_{l,j}|$ & - count of word $w_l$ in document $D_j$\\ 
  $|D_j|$&- total number of words in the document $D_j$\\
  $\mathcal{P}$ & - the topic distribution for $D_j$ from topic model $\mathcal{T}$\\
  $\mathcal{P}'$ & - the topic distribution for $D_j$ from topic model $\mathcal{T}'$\\
  $\mathcal{Q}_i$ & - the word distribution for topic $T_i$ in $\mathcal{T}$\\ 
  $\mathcal{Q}'_i$ & - the word distribution for topic $T'_i$ in $\mathcal{T}'$\\
  $p_{i,j}$ & - probability of topic $\mathcal{T}_i$ occurring in document $D_j$ output  \\
  $q_{i,l}$ &- probability of word  $w_l$ occurring in topic $\mathcal{T}_i$  \\ 
  $q'_{i,l}$ &- probability of word  $w_l$ occurring in topic $\mathcal{T}'_i$ \\
  $p'_{i,j}$ & - probability of topic $\mathcal{T}'_i$ occurring in document $D_j$ output \\
  $T(w)$ & - probability of word $w$ occurring in topic $T$\\
  $T'(w)$ & - probability of word $w$ occurring in topic $T'$\\
\hline\end{tabular}}
 \label{Notations}
\end{table}

\begin{definition}The total variation distance between two probability distributions $P$ and $Q$ is defined as \begin{gather}\delta_{TV}(P,Q)=sup_{A\in \mathcal{F}}|P(A)-Q(A)|\end{gather} where $\mathcal{F}$ represents a sigma-algebra on the subset of the sample space $\Omega$.\end{definition}
\begin{theorem} For a countable set, $\Omega$ \begin{gather}\delta_{TV}(P,Q)=\frac{1}{2}||P-Q||_1=\frac{1}{2}\sum_{\omega \in \Omega}|P(\omega)-Q(\omega)|\end{gather}\end{theorem} \begin{assumption}If a topic $p$ has nonzero probability of occurring in a document, then the topic must contribute at least one count for each word in it for the document. \label{assum1} \end{assumption} 
\begin{lemma}From assumption \ref{assum1}, \begin{gather}min_{i,j}\{p_{i,j}\}\geq \frac{k}{D_{max}}\end{gather}\label{pmin}\end{lemma}\begin{proof}Restating assumption \ref{assum1} we get,
\begin{gather*}\forall i, \forall j, \forall l \thinspace\thinspace i \in [t], j \in [n], l \in [|\mathcal{V}|]\\ p_{i,j}\cdot q_{i,l} \geq \frac{1}{|D_j|}\\ \Rightarrow p_{i,j}\cdot q_{i,l} \geq \min_{j}\Big\{\frac{1}{|D_j|}\Big\}\\\Rightarrow p_{i,j}\cdot q_{i,l} \geq \frac{1}{|D_{max}|}\\\Rightarrow p_{i,j}\cdot min_{l}\{q_{i,l}\} \geq \frac{1}{|D_{max}|} \\\Rightarrow p_{i,j} \geq \frac{k}{|D_{max}|}[\because \min_l\{q_{i,l}\} \leq \frac{1}{k}]  \end{gather*} Thus $\min_{i}\{p_{i,j}\}\geq \frac{k}{|D_{max}|}$\end{proof}
\begin{lemma}The maximum possible value of $p_{max}$ is $1-(t-1)\cdot\frac{k}{|D_{max}|}$.\label{pmax}\end{lemma}\begin{proof} According to the proof statement, the topic mixture for a document $D_j$ is given by $\mathcal{P}_j=\langle \frac{k}{|D_{max}|}, \cdots,\frac{k}{|D_{max}|},1-(t-1)\cdot \frac{k}{|D_{max}|}\rangle$. We will prove this by contradiction. Let there exist some $\bar{p} \in \mathcal{P}_j$ such that $\bar{p} >  \frac{k}{|D_{max}|}$. Clearly this means that $\max_{i,j}\{p_{i,j}\} = 1-(t-2)\cdot \frac{k}{|D_{max}|}-\bar{p} < 1-(t-1)\cdot \frac{k}{|D_{max}|}$. Clearly from lemma \ref{pmin} this concludes our proof.   \end{proof}
\begin{theorem}For any pair of topics $(\mathbb{T},\mathbb{T}') \in \mathcal{T}\times \mathcal{T}'$, 
\begin{gather*}
{\textstyle
||\mathbb{T}-\mathbb{T}'||_1 \geq 2  \frac{1}{\Big(1-(t-1) \frac{k}{\max_j|D_j|}\Big)}  \Big(\frac{\mathcal{C}_{min}}{t} -\frac{1}{2} \Big(1-t \frac{k}{\max_j|D_j||}\Big)\Big),} 
\end{gather*} 

where $\mathcal{C}_{min}=min_{j,l}\Big\{\frac{v\cdot(|D_j|-|w_{l,j}|\omega_j)}{|D_j|\cdot(|D_j|+v\cdot \omega_j}\Big\}$, $|D_j|$  is the total number of words in $D_j$, $\omega_j$ is the total number of  unique words, $v$ is the variance of the distribution 
$Lp(\epsilon', \delta',d),$\\$ v = p r e^{(\epsilon'\cdot\eta_0)/d}\big(\frac{1}{(1-r)^2}+\frac{2r}{(1-r)^3}\big)+ p\bar{f} r \big(e^{-(\epsilon'\cdot\eta_0)/d}-e^{(\epsilon'\eta_0)/d}\big) - \big(pe^{(\epsilon'\cdot\eta_0)/d}\frac{r}{(1-r)^2}+ p f  (e^{-(\epsilon'\cdot\eta_0)/d}-e^{(\epsilon'\eta_0)/d})\big)^2,$ 
$ \bar{f}=\frac{df}{dr}, $ 
$f=r\big(\frac{1-r^d-d\cdot r^{d-1}(1-r)}{(1-r)^2}\big),$ 
$ r=e^{-(\epsilon'/d)},  p=\frac{e^{\epsilon/d}-1}{e^{\epsilon/d}+1},$ 
$\eta_0 = -\frac{d \cdot \ln((e^{\epsilon'/d} + 1)\delta')}{\epsilon} + d, \epsilon'=ln(1+\frac{1}{\beta}(e^\epsilon-1)), $
$\delta'=\beta\delta$ and $|w_{l,j}|$ is the number of times the word 
 $w_l \in \mathcal{V}$ appears in transcript $D_j$.\end{theorem} 
 
 \begin{proof} For any word $w_l$ in a document $D_j, l \in [|\mathcal{V}|], j \in [n]$, \begin{gather}\sum_{i=1}^{t}p_{i,j}q_{i,l} =\frac{|w_{l}|}{|D_j|}\label{true}\\\sum_{i=1}^t p'_{i,j}q'_{i,l}=\frac{|w_l|+v}{|D_j|+v\cdot\omega_j}  \label{perturbed}\end{gather}
Now subtracting eq \ref{true} from eq \ref{perturbed} we  \begin{gather}\sum_{i=1}^t( p'_{i,j}\cdot q'_{i,l}-p_{i,j}\cdot q_{i,l})=\frac{v\cdot(|D_j|-|w_l|\omega_j)}{|D_j|\cdot(|D_j|+v\cdot\omega_j)} \end{gather}
Let $\mathcal{C}_{j,l}=\frac{v\cdot(|D_j|-|w_l|\omega_j)}{|D_j|\cdot(|D_j|+v\cdot\omega_j)}$.
\begin{gather}\sum_{i=1}^t| p'_{i,j}\cdot q'_{i,l}-p_{i,j}\cdot q_{i,l}| \geq \sum_{i=1}^t( p'_{i,j}\cdot q'_{i,l}-p_{i,j}\cdot q_{i,l}) = \mathcal{C}_{j,l}\end{gather}
Now, observe that $min\{max_{i}\{|p'_{i,j}\cdot q'_{i,l}-p_{i,j}\cdot q_{i,l}|\}\} $ occurs when $|p'_{1,j}\cdot q'_{1,l}-p_{1,j}\cdot q_{1,l}|=\cdots=|p'_{t,j}\cdot q'_{t,l}-p_{t,j}\cdot q_{t,l}|\geq \frac{\mathcal{C}_{j,l}}{t}$.
Let $p \in \mathcal{P}_j, p' \in \mathcal{P}_j', q \in \mathcal{Q}_i$ and $q' \in \mathcal{Q}_i'$ such that $\mathcal{P},\mathcal{P}',\mathcal{Q},\mathcal{Q}'$ corresponds to $min\{max_{i}\{|p'_{i,j}\cdot q'_{i,l}-p_{i,j}\cdot q_{i,l}|\}\} $. 
Now renaming as follows \begin{gather}q_1=max\{q,q'\}\\p_1=\left\{ \begin{array}{ll}
         p & \mbox{if $q_1=q$};\\
        p' & \mbox{otherwise}.\end{array} \right. 
        \\
        q_2=min\{q,q'\}\label{q2}\\p_2=\left\{ \begin{array}{ll}
         p & \mbox{if $q_2=q$};\\
        p' & \mbox{otherwise}.\end{array} \right.
        \end{gather}
        We get,
\begin{gather*}|p_1\cdot q_1-p_2\cdot q_2|\geq \frac{\mathcal{C}_j}{t}\\\Rightarrow |(q_1-q_2)\cdot p_1+q_2\cdot(p_1-p_2)| \geq \frac{\mathcal{C}_j}{t}\\\Rightarrow |(q_1-q_2)\cdot p_1|+|q_2\cdot(p_1-p_2)| \geq \frac{\mathcal{C}_{j,l}}{t}\\\Rightarrow |(q_1-q_2)\cdot p_1|+q_2\cdot \Big(1-t\cdot\frac{k}{|D_{max}|}\Big)\geq \frac{\mathcal{C}_j}{t} [\because \mbox{From lemma \ref{pmax}}] \\\Rightarrow |(q_1-q_2)\cdot p_1|+\frac{1}{2}\cdot \Big(1-t\cdot\frac{k}{|D_{max}|}\Big)\geq \frac{\mathcal{C}_{j,l}}{t}[\because \mbox{ By eq \ref{q2}} q_2 < \frac{1}{2}] \\ 
\Rightarrow (q_1-q_2)\cdot \Big(1-(t-1)\cdot\frac{k}{|D_{max}|}\Big)\geq \frac{\mathcal{C}_{j,l}}{t} -\frac{1}{2}\cdot \Big(1-t\cdot\frac{k}{|D_{max}|}\Big)\\ [\because \mbox{ By  lemma \ref{pmax}}]\\\Rightarrow q_1-q_2 \geq \frac{1}{\Big(1-(t-1)\cdot\frac{k}{|D_{max}|}\Big)}\cdot \Big(\frac{\mathcal{C}_{j,l}}{t} -\frac{1}{2}\cdot \Big(1-t\cdot\frac{k}{|D_{max}|}\Big)\Big)
\end{gather*}

Now clearly
\begin{gather}min_{j,l}\{C_{j,l}\}=max_j\{min_l\{\frac{\frac{2 d}{\epsilon}\cdot(|w_l|\omega_j-|D_j|)}{|D_j|\cdot(|D_j|+\frac{2 d \cdot \omega_j}{\epsilon})}\}\}\end{gather}
Let $\mathcal{C}_{min} = min_{j,l}\{\mathcal{C}_{j,l}\}$. Thus for any pair $(T,T') \in \mathcal{T}\times \mathcal{T}'$ we have \begin{gather*}\delta_{TV}(T,T') \geq sup_{w \in \mathcal{V}}|T(w)-T'(w)| \\\geq \frac{1}{\Big(1-(t-1)\cdot\frac{k}{|D_{max}|}\Big)}\cdot \Big(\frac{\mathcal{C}_{min}}{t} -\frac{1}{2}\cdot \Big(1-t\cdot\frac{k}{|D_{max}|}\Big)\Big)\numberthis\label{delta} \end{gather*}
Now from theorem 2.1, we have
\begin{gather}||T-T'||_1 \geq 2\cdot \frac{1}{\Big(1-(t-1)\cdot\frac{k}{|D_{max}|}\Big)}\cdot \Big(\frac{\mathcal{C}_{min}}{t} -\frac{1}{2}\cdot \Big(1-t\cdot\frac{k}{|D_{max}|}\Big)\Big)\end{gather}
\end{proof}

\section{Appendix B}
\subsection{Sensitive Keywords Lists}

Table~\ref{tab:keywordLists} presents the sensitive keywords list for each of the evaluation datasets. We discuss \name's efficacy in scrubbing them and its effect on the transcription utility in Sec.~\ref{sec:eval:sws}.

\begin{table}[ht]
\small
\centering
\ra{1}
\scalebox{0.70}{\begin{tabular}{@{}clllll@{}}\toprule
\textbf{Datasets} & \textbf{Facebook 1} & \textbf{Facebook 2} & \textbf{Facebook 3} & \textbf{Carpenter 1} & \textbf{Carpenter 2} \\
\midrule
\multirow{7}{*}{\rotatebox[origin=c]{90}{\textbf{keywords}}}  &  South Carolina & Grassley & Obama  & Miller & Wessler\\
& Harvey & Cambridge Analytica & Tillis & Smith & Massachusetts \\
& Sandberg & South Dakota & Harris  & Stewart & Kennedy \\
& Palo Alto &  & Kogan  & kogan & Carpenter\\
& Cheryl & & Alexander  & Greenwood & Graham \\
&  & &   & Roberts  & Kogan \\
&  & &  & Breyer & \\
\end{tabular}}
\caption{Sensitive keywords lists for each dataset}
\label{tab:keywordLists}
\end{table}

\subsection{NLP Generated Text}
In this section, we show a piece of the NLP generated text using the GPT2 language model. The model's input is a few sentences from the Carpenter case.

\textit{"In determining that the information is reasonable, the Court may find that this amount of intrusion is reasonable and that it is a reasonable effort by the government to harass privacy rights protections that normally apply to American citizens. Have there been any other issues before the Court before and after this hearing that we have not yet addressed. We will be looking again, on appeal. Would the government have done so differently today if they had arrested Petitioner five years earlier? The Court agrees that Petitioner has the right to petition the Court. I would agree with the court's conclusion that the government had previously attempted to arrest Petitioner, such as by conducting warrantless searches of his phone and computer in October 2009. There are some questions on the record, however, that we would have looked into in more detail in that regard.}

\textit{DAVID GREENE: Well, they seem to have two thoughts that I am not going to sit down and define as true or false; I think -- I've said it before. One thing we may know about this case is that the judge, Judith Miller, is a professor in Duke University's criminal-justice faculty and when she refused to reach a plea deal, the defense asked her to seek a final decision in the case. They considered a reduced sentence because this is his plea deal, which means he's allowed to serve his entire life in prison."}

\subsection{Segmentation Algorithm}
Here, we outline the formal segmentation algorithm discussed in Sec.~\ref{sec:Segmentation}

\begin{algorithm}
\caption{Hierarchical Speech Segmentation}\label{algo:segmentation}
\begin{algorithmic}[1] 
\Statex \textbf{Input}: Speech file $S$
\Statex \hspace{1cm}                       Minimum segment duration $l_{\mathbb{S}}$
\Statex \textbf{Output}: Sequence of speech segments $\mathbb{S}$
\Statex \textbf{Stage 1}: Silence Detection
\State Identify the timestamps $t_{S}=\{t^i_{S}\}$ corresponding to periods of silence in $S$ 
\STATE  Divide $S$ into a sequence of coarse segments $\mathbb{S}_1$ such that each segment is bounded by sequential timestamps from $t_S$
\Statex \textbf{Stage 2}: Pitch Detection
\STATE $\mathbb{S}=\varnothing$
\State \textbf{for} segment $\mathbb{S}^i_1 \in \mathbb{S}_1  $ 
\State \hspace{0.5cm} Divide each $\mathbb{S}^i_1$ into a sequence of finer segments, $\mathbb{S}^i_
2$, \Statex \hspace{0.5cm}by identifying  glottal cycles with buffers of \textit{non-}\Statex \hspace{0.5cm}\textit{speech} at the segment boundaries
 \State \hspace{0.5cm} \textbf{do} 
 \STATE \hspace{0.8cm} merge adjacent segments in $\mathbb{S}^i_2$ into a longer \Statex \hspace{0.8cm} segment $\mathbb{S}^i$   \State \hspace{0.5cm}\textbf{while} $(length(\mathbb{S}^i) < l_{\mathbb{S}})$
 \STATE \hspace{0.5cm} $\mathbb{S}=\mathbb{S} + \mathbb{S}^i$
 \hspace{1cm}\textcolor{blue}{$\rhd$} ``$+$'' denotes that the segments \Statex \hspace{3.2cm} are added in a sequence ordered \Statex \hspace{3.2cm} by  their timestamps
\STATE \textbf{end for}
\STATE Return $\mathbb{S}$ 
\end{algorithmic}
\end{algorithm}

\fi
\end{document}